\documentclass[letterpaper,10pt]{IEEEtran}
\IEEEoverridecommandlockouts

\usepackage[colorlinks=true]{hyperref}
\hypersetup{urlcolor=blue,linkcolor=blue,citecolor=blue,colorlinks=true}
\usepackage{graphicx,amsmath,amssymb,arydshln,color,cite}
\usepackage[space]{grffile}
\usepackage{relsize}
\usepackage[Symbol]{upgreek}
\usepackage{aurical}
\usepackage{subcaption}
\usepackage{dsfont}
\allowdisplaybreaks

\graphicspath{{./Figures/}}

\newtheorem{definition}{\textbf{Definition}}
\newtheorem{theorem}{\textbf{Theorem}}

\newtheorem{lemma}{\textbf{Lemma}}[subsection]
\newtheorem{assumption}{\textbf{Assumption}}
\newtheorem{proposition}{\textbf{Proposition}}
\newtheorem{remark}{\textbf{Remark}}[section]
\newtheorem{innercustomass}{\textbf{Assumption}}
\newenvironment{customass}[1]
{\renewcommand\theinnercustomass{#1}\innercustomass}
{\endinnercustomass}

\newcommand{\systembd}[9]{
	\left[
	\begin{array}{c|cc}
		#1 & #2 & #3   \\
		\hline
		#4 & #5 & #6   \\
		#7 & #8 & #9
	\end{array}
	\right]
}
\newcommand{\pnorm}{\mathbf \Delta}
\newcommand{\gvec}{\boldsymbol{\gamma}}
\newcommand{\ito}{\diamond_{\mathsmaller I}}
\newcommand{\strat}{\diamond_{\mathsmaller S}}
\newcommand{\cov}[1]{\mathbf #1}
\newcommand{\sscov}[1]{\bar{\mathbf #1}}
\newcommand{\ecov}[1]{\hat{\mathbf #1}}
\newcommand{\expec}[1]{\mathbb E \left[#1\right]}
\newcommand{\mnorm}[1]{\left|\left| #1\right|\right|}
\newcommand{\vnorm}[1]{\left|\left| #1\right|\right|}
\newcommand{\tr}[1]{\text{tr}\left( #1 \right)}

\newcommand{\tv}[2]{\mathcal T \mathcal V_0^{#1} \left(#2 \right)}
\newcommand{\qv}[2]{\mathcal Q \mathcal V_0^{#1} \left(#2 \right)}
\newcommand{\sqrtp}[1]{\left(#1\right)^{\frac{1}{2}}}
\newcommand*{\Scale}[2][4]{\scalebox{#1}{$#2$}}%
	
\title{\LARGE \bf
	An Input-Output Approach to \\ Structured Stochastic Uncertainty in Continuous Time
}

\author{Maurice Filo and Bassam Bamieh
	\thanks{This work is supported by NSF Awards ECCS-1408442.}
	\thanks{Maurice Filo and Bassam Bamieh are with the Department of Mechanical Engineering, University of California, Santa Barbara,
		Santa Barbara, California 93117, USA.       {\it\small filo@umail.ucsb.edu, bamieh@engineering.ucsb.edu}
}	}

\begin{document}
\maketitle
\thispagestyle{empty}
\pagestyle{empty}

\begin{abstract}
	We consider the continuous-time setting of linear time-invariant (LTI) systems in feedback with multiplicative stochastic uncertainties. The objective of the paper is to characterize the conditions of Mean-Square Stability (MSS) using a purely input-output approach, i.e. without having to resort to state space realizations. This has the advantage of encompassing a wider class of models (such as infinite dimensional systems and systems with delays). The input-output approach leads to uncovering new tools such as stochastic block diagrams that have an intimate connection with the more general Stochastic Integral Equations (SIE), rather than Stochastic Differential Equations (SDE). Various stochastic interpretations are considered, such as It\=o and Stratonovich, and block diagram conversion schemes between different interpretations are devised. The MSS conditions are given in terms of the spectral radius of a matrix operator that takes different forms when different stochastic interpretations are considered.
\end{abstract}

\section{INTRODUCTION}
Linear Time-Invariant (LTI) systems with stochastic disturbances is a powerful modeling technique that is used to analyze and control a large class of physical systems. While additive disturbances are most commonly used to model process and measurement noise in a system, multiplicative disturbances are often necessary to model stochastic uncertainties in the system parameters (such as coefficients in dynamical equations). LTI systems driven by additive stochastic processes are more common in the literature; whereas simultaneous additive and multiplicative disturbances are relatively less addressed. The present paper develops a methodology to study the mean-square stability of continuous-time systems with both additive and multiplicative disturbances, while adopting different stochastic interpretations (such as It\=o and Stratonovich). 
\begin{figure}[h!]
	\centering
	\begin{tabular}{ll}
		\includegraphics[scale = .67]{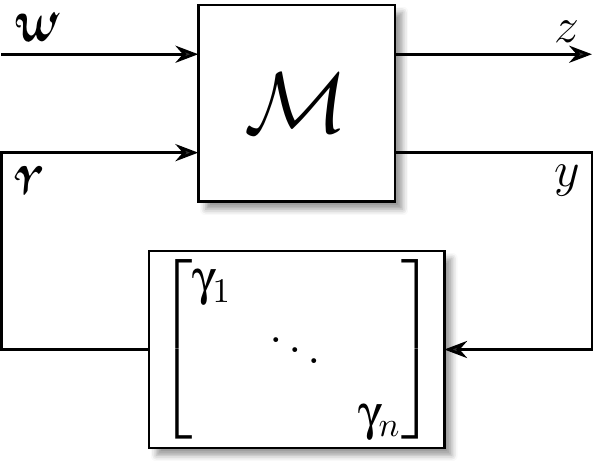} & 
		\includegraphics[scale = .67]{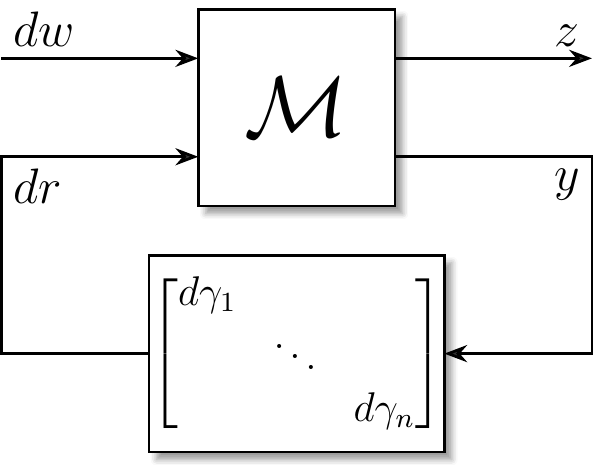} \\
		\footnotesize{(a) White Process Representation} & \footnotesize{(b) Wiener Process Representation}
	\end{tabular}
	\caption{\footnotesize{The general continuous-time setting of linear systems with both additive and multiplicative stochastic disturbances. Both block diagrams describe the same setting, given in (\ref{Eqn: White Process Feedback}) and (\ref{Eqn: Wiener Process Feedback}), using white processes (to the left) and Wiener processes (to the right), respectively. The LTI system $\mathcal M$ is in feedback with multiplicative stochastic gains represented here as a diagonal matrix. In Figure~(a), {\Fontauri{w}} is an additive stationary white process, while $\upgamma_1, \cdots, \upgamma_n$ are multiplicative stationary white processes. In Figure~(b), $dw$ represents the differential of an additive Wiener process, while $d\gamma_1, \cdots, d\gamma_n$ represent the differentials of (possibly correlated) Wiener processes that enter the dynamics multiplicatively. The  signal $z$ represents an output whose variance quantifies a performance measure.}   }
	\label{Fig: Concept Block Diagram}
\end{figure}

The general setting we consider in this paper is the continuous-time analog of that presented in \cite{bamieh2018structured} and is depicted in Figure~\ref{Fig: Concept Block Diagram}(a). An LTI system is in feedback with  stochastic gains $\upgamma_1(t), ... \upgamma_n(t)$, that are assumed to be ``white" in time (i.e. temporally independent) but possibly mutually correlated. Another set of stochastic disturbances are represented by the vector-valued signal {\Fontauri w} which is also assumed to be white but enters the dynamics additively. The signal $z$ is an output whose variance quantifies a performance measure. The feedback term is then a diagonal matrix with the individual gains $\{\upgamma_i\}$ appearing on the diagonal. Such gains are commonly referred to as structured uncertainties. Note that if the gains are deterministic (but uncertain), we obtain the general setting considered in the robust control literature (e.g. \cite{zhou1996robust}). The main objective of the present paper is to derive the necessary conditions of Mean-Square Stability (MSS) for systems taking the form of Figure~\ref{Fig: Concept Block Diagram}(a). The treatment is carried out using a purely input-output approach (i.e. without giving $\mathcal M$ a state space realization). This has the advantage of encompassing a wider class of models $\mathcal M$ (e.g. infinite dimensional systems).

In a discrete-time setting, there is no ambiguity of defining white (i.e. temporally independent) signals. However, in a continuous-time setting, technical issues arise because white signals are not mathematically well defined when they enter the dynamics multiplicatively. Hence, the block diagram in Figure~\ref{Fig: Concept Block Diagram}(a) is only used to pose the problem setup in an analogous fashion to the discrete-time setting in \cite{bamieh2018structured}, but at the cost of abandoning mathematical rigor. In fact, the equations describing Figure~\ref{Fig: Concept Block Diagram} can be written using the white processes {\Fontauri w} and $\{\upgamma_i\}$ as
\begin{align} \label{Eqn: White Process Feedback}
\begin{bmatrix} z \\ y \end{bmatrix} &= \mathcal M \begin{bmatrix} \text{\Fontauri w} \\ \text{\Fontauri r} \end{bmatrix} \Leftrightarrow 
\begin{bmatrix} z(t) \\ y(t) \end{bmatrix} = \mathlarger{\int_{0}^{t}} M(t-\tau)
\begin{bmatrix} \text{\Fontauri w}(\tau) \\ \text{\Fontauri r}(\tau) \end{bmatrix} d\tau \nonumber \\
\text{\Fontauri r}(t) &= \mathcal D \big( \boldsymbol \upgamma(t)\big) y(t), 
\end{align}
where $M$ is the impulse response of $\mathcal M$, and $\mathcal D \big( \boldsymbol \upgamma(t)\big)$ is a diagonal matrix whose elements are equal to those of $\boldsymbol \upgamma(t): = \begin{bmatrix} \upgamma_1(t) & \cdots &  \upgamma_n(t) \end{bmatrix}^*$. To resort back to mathematical rigor, we think of the white processes {\Fontauri w} and $\{\upgamma_i\}$ as the formal derivatives of Wiener processes (or Brownian motion) that are mathematically well defined \cite{oksendal2003stochastic}. More precisely, define
\begin{align} \label{Eqn: Increments Definition}
\upgamma_i(t) := \frac{d\gamma_i(t)}{dt}; \quad \text{\Fontauri w}(t)  := \frac{dw(t)}{dt}; \quad \text{\Fontauri r}(t) := \frac{dr(t)}{dt}, 
\end{align}
such that $\boldsymbol \gamma(t) := \begin{bmatrix} \gamma_1(t) & \cdots & \gamma_n(t) \end{bmatrix}^*$ and $w(t)$ represent nonstandard, vector-valued Wiener processes (i.e. their covariances do not have to be the identity matrix). Furthermore, $r(t)$ will be shown (Section~\ref{Section: Independence of increments of r}) to have temporally independent increments when $\mathcal M$ is causal and the It\=o interpretation is adopted. Hence, the equations can be rewritten using differential forms as
\begin{align} \label{Eqn: Wiener Process Feedback}
&\begin{bmatrix} z \\ y \end{bmatrix} = \mathcal M \begin{bmatrix} dw \\ dr \end{bmatrix} \Leftrightarrow 
\begin{bmatrix} z(t) \\ y(t) \end{bmatrix} = \mathlarger{\int_{0}^{t}} M(t-\tau)
\begin{bmatrix} dw(\tau) \\ dr(\tau) \end{bmatrix} \nonumber\\
&dr(t) = \mathcal D \big(d{\boldsymbol \gamma}(t)\big) y(t). 
\end{align}
These equations are now mathematically well defined when given some desired interpretation such as in the sense of It\=o or Stratonovich. It will be shown in Section~\ref{Section: MSS Conditions} that different interpretations produce different conditions of MSS.

We should note the other common and related models in the literature which are usually done in a state space setting and can be represented as Stochastic Differential Equations (SDEs). One such model is a linear system with a random ``A matrix" such as  
\begin{align} \label{Eqn: Random A Matrix} \dot x(t) = A(t) x(t)  + B \text{\Fontauri w}(t), \end{align}
where $A(t)$ is a matrix-valued stochastic process independent of $\{x(\tau), \tau \leq t\}$. One can always rewrite $A(t)$ in terms of scalar-valued stochastic processes so that
$$\dot x(t) = \big(A_0 + \upgamma_1(t) A_1 + \cdots + \upgamma_n(t) A_n\big) x(t) + B \text{\Fontauri w}(t).$$
If the matrices $A_1,\ldots,A_n$ are all of rank 1 (e.g. $A_i=b_ic_i$, for column and row vectors $b_i$, $c_i$  respectively, $i=1,\ldots,n$), then it is well-known~\cite{zhou1996robust} that the  model~(\ref{Eqn: Random A Matrix}) can always be reconfigured like the block diagram of Figure~\ref{Fig: Concept Block Diagram}(a) by setting 
$$ \mathcal M ~=~ \systembd{A_0}{B}{B_0}{C}{0}{0}{C_0}{0}{0}, $$
where $B_0:=\begin{bmatrix} b_1 & \cdots & b_n \end{bmatrix}$ and $C_0:= \begin{bmatrix} c_1^* & \cdots & c_n^* \end{bmatrix}^*$. In the example above, we have chosen $z=Cx$. If the matrices $\{A_i\}_{i=1}^n$ are not rank one,  it is still possible to reconfigure~(\ref{Eqn: Random A Matrix}) into a diagram like Figure~\ref{Fig: Concept Block Diagram}(a), but with the perturbation blocks being ``repeated''~\cite{packard1988structured}. 

When the processes $\{\upgamma_i\}$ and {\Fontauri w} are ``white" in time, we resort to the configuration of Figure~\ref{Fig: Concept Block Diagram}(b) to express the stochastic disturbances in terms of Wiener processes. Exploiting (\ref{Eqn: Increments Definition}) yields
\begin{align} 
\mathcal M: &\left\{
\begin{aligned}
dx(t)& = A_0x(t)dt + B_0dr(t) + B dw(t) \\
y(t) &= C_0x(t) \\	
z(t) &= C x(t)			
\end{aligned}
\right.  \label{Eqn: State SDE}\\
&dr(t)= \mathcal D\big(d\boldsymbol \gamma(t)\big) y(t).  \label{Eqn: Feedback SDE}
\end{align}
Equations (\ref{Eqn: State SDE}) and (\ref{Eqn: Feedback SDE}) describe the block diagram of Figure~\ref{Fig: Concept Block Diagram}(b) when $\mathcal M$ is given as a state space realization. In fact, the impulse response can be easily calculated to be
$$ M(t) := \begin{bmatrix} C \\ C_0 \end{bmatrix} e^{A_0 t} \begin{bmatrix} B & B_0 \end{bmatrix},$$
thus showing that models like those given in (\ref{Eqn: Random A Matrix}) are a special case of the purely input-output approach that we consider in this paper. On a side note, observe that the underlying stochastic dynamics of the state $x$ in (\ref{Eqn: State SDE}) and (\ref{Eqn: Feedback SDE}) can be rewritten in a single SDE, that involves both additive and multiplicative disturbances, as
\begin{align} \label{Eqn: Single SDE}
dx(t) = A_0x(t)dt + B_0 \mathcal D\big(Cx(t)\big) d\boldsymbol \gamma(t) + B dw(t).
\end{align}
Particularly, \cite{el1992stability} studied SDEs having the form of (\ref{Eqn: Single SDE}) interpreted in the sense of It\=o, where $B=0$ (i.e. no additive noise) and $\boldsymbol \gamma$ is ``spatially uncorrelated", i.e. $\mathbb E[\gamma_i \gamma_j] = 0, \forall i\neq j$.

Our goal in this paper is to extend the machinery developed in \cite{bamieh2018structured} to provide a rather elementary, and purely input-output treatment and derivation of the 
necessary and sufficient conditions of MSS for systems like that of Figure~\ref{Fig: Concept Block Diagram}. Furthermore, our treatment covers both It\=o and Stratonovich interpretations. It is shown that the conditions of MSS can be stated in terms of the spectral radius of a finite dimensional linear operator defined in Section~\ref{Section: MSS Conditions}. It is also shown that this operator takes different forms when different stochastic interpretations are prescribed (such as It\=o or Stratonovich).  

The paper is organized as follows. First we provide some useful definitions and notation. Then, in Section~\ref{Section: Problem}, we give a precise formulation of the problem statement by setting up a general ``stochastic block diagram" and describing the underlying assumptions. In Section~\ref{Section: Results}, we present the main results of the paper that can be divided into two parts. The first part shows a block diagram conversion scheme from Stratonovich to It\=o interpretations, and the second part states the conditions of mean-square stability. The special cases of state space realizations are then treated in Section~\ref{Section: SS}.  Sections~\ref{Section: MS Equivalence} and~\ref{Section: LGO and MSS} provide the detailed derivations that explain the results. Finally, we conclude in Section~\ref{Section: Conclusion}.

\section{Preliminaries and Notation} \label{Sec: Preliminaries}
All the signals considered in this paper are  defined on the semi-infinite, continuous-time interval $\mathbb R^+ := [0,+\infty)$. The dynamical systems considered are maps between various signal spaces over the time interval $\mathbb R^+$.
Unless stated otherwise, all stochastic processes in this paper are random vector-valued functions of (continuous) time.

\subsection*{Notation Summary}
\subsubsection{\textbf{Variance \& Covariance Matrix of a Signal}} 
If $v$ is a stochastic signal, then its instantaneous variance and covariance matrix are denoted by the lowercase and uppercase bold letters respectively
$$ \cov v(t) := \expec{v^*(t)  v(t) } \quad \text{and} \quad \cov V(t)  := \expec{v(t)  v^*(t)}, $$
where $v^*$ denotes the transpose of $v$. The entries of $\cov V(t)$ are the mutual correlations of the vector $v(t)$, and are sometimes referred to as \textit{spatial correlations}. Note that $\tr{\cov V(t)} = \cov v(t)$.
\subsubsection{\textbf{Variance \& Covariance Matrix of a Differential Signal}}
If the differential $du$ of a stochastic signal $u$ appears in a stochastic block diagram (see Figure~\ref{Fig: MSS Setting} for example), its instantaneous variance and covariance are represented as
$$ \expec{du^*(t) du(t) }:= \cov u(t) dt \quad \text{and} \quad  \expec{du(t) du^*(t)} := \cov U(t) dt, $$
respectively. This is a compact (differential) notation for
$$ \expec{u^*(t) u(t)}:= \int_0^t \cov u(\tau) d\tau; \quad \expec{u(t) u^*(t)} := \int_0^t \cov U(\tau) d\tau. $$
\subsubsection{\textbf{Steady State Variance \& Covariance Matrix}}
The asymptotic limits of the instantaneous variance and covariance matrix, when they exist, are denoted by an overbar, i.e.
$$\sscov u:= \lim_{t\to\infty} \cov u(t) \quad \text{and} \quad \sscov{U} := \lim_{t\to\infty} \cov U(t).$$
\subsubsection{\textbf{Second Order Process}} A process $v$ is termed \textit{second order} if the entries of its covariance matrix, $\cov V(t)$, are finite for each $t\in \mathbb R^+$. 
\subsubsection{\textbf{Probability Space}}
Let $(\Omega, \mathcal F, p)$ be a complete probability space with $\Omega$ being the sample space, $\mathcal F$ the associated $\sigma-$algebra and $p$ the probability measure. Let $L_2(p)$ denote the space of vector-valued random variables with finite second order moments. Note that $L_2(p)$ is a Hilbert space. 
\subsubsection{\textbf{Equalities \& Limits in the Mean-Square Sense}}
Two stochastic processes $x$ and $y$ are said to be equal in the mean-square sense if
$\expec{\vnorm{x-y}^2} = 0$, where throughout the paper $\vnorm{.}$ denotes the $\ell^2-\text{norm}$ for vectors and the spectral norm for matrices. 

A  sequence of second order stochastic processes, $\{x_N\}$, is said to converge to $\bar x \in L_2(p)$ in the mean-square sense iff $\lim_{N\to\infty} \vnorm{x_N - \bar x}^2 = 0$.
\subsubsection{\textbf{White Process}} A stochastic process $\upgamma$ is termed \textit{white} if it is uncorrelated at any two distinct times, i.e. $ \expec{\upgamma(t) \upgamma^*(\tau)}= \cov \Gamma \delta(t-\tau)$, where $\delta$ is the Dirac delta function. Note that in the present context, a white process $\upgamma$ may still have spatial correlations, i.e. its instantaneous covariance matrix $\cov \Gamma$ need not be the identity. 
\subsubsection{\textbf{Vector-Valued Wiener Process}} \label{Section: Wiener Process} In a continuous-time setting, calculus operations on a white process entering the dynamics multiplicatively are not mathematically well defined. Hence, it is useful to represent a white process as the formal derivative of a Wiener process, i.e. $\upgamma(t) := \frac{d\gvec(t)}{dt}$, where $\gvec$ is a zero-mean, vector-valued Wiener process with an instantaneous covariance matrix $\expec{\gvec(t) \gvec^*(t)} = \cov \Gamma t$. This can be equivalently written in differential form as $\expec{d\gvec(t) d\gvec^*(t)} = \cov \Gamma dt$. Note that $\gvec$ is said to have temporally independent increments, i.e. its differentials $\big(d\gvec(t), d\gvec(\tau) \big)$ are independent when $t\neq \tau$.
\subsubsection{\textbf{Partitions of Time Intervals}} \label{Section: Partition}
Let $\mathcal P_N[0,t]$ denote an arbitrary partition of the time interval $[0,t]$ into $N$ subintervals $[t_k,t_{k+1}]$ for $k=0, 1, \cdots, N-1$, such that $0=t_0<t_1<\cdots<t_N=t$. The partition step-size is denoted by $\Delta_k:= t_{k+1}- t_k$ and the norm of the partition $\mathcal P_N[0,t]$ is denoted by the bold letter $\cov \Delta$ defined as
$ \cov \Delta := || \mathcal P_N[0,t] || = \sup_k\Delta_k.$
Note that $\lim_{N\to\infty}{\cov \Delta} = 0$. 
\subsubsection{\textbf{Notation for Signals and Increments on $\mathcal P_N[0,t]$}} \label{Section: Increments}
With slight abuse of notation, a continuous-time stochastic signal $\{u(\tau), 0\leq \tau \leq t\}$ is represented at node $t_k$ of the partition $\mathcal P_N[0,t]$ as $u_k := u(t_k)$ for $k = 0, 1, \cdots, N$. The increments of $\{u(\tau), 0\leq\tau\leq t\}$ at $t_k$ are denoted by $\tilde u_k := u(t_{k+1}) - u(t_k)$ for $k=0,1,\cdots, N-1$, and they represent a finite approximation of the differential form $\{du(\tau), 0 \leq \tau \leq t \}$. 

A continuous-time stochastic process $u$ is said to have \textit{temporally independent increments} if $\big(du(t), du(\tau)\big)$ are independent whenever $t\neq \tau$. This implies that, on the partition $\mathcal P_N[0,t]$, $(\tilde u_k, \tilde u_l)$ are independent whenever $k\neq l$. 
\subsubsection{\textbf{Stochastic Integrals}} \label{Section: Interpretations} Calculus operations on a Wiener process are mathematically well defined when some stochastic interpretation is prescribed (such as It\=o or Stratonovich). Particularly, we distinguish It\=o and Stratonovich integrals using the symbols "$\ito$" and "$\strat$", respectively. More precisely, let $v$ be a vector-valued second order stochastic process and $\gvec$ be a vector-valued Wiener process. If $\Gamma(t) := \mathcal D\big(\gvec(t)\big)$ is a diagonal matrix whose entries are equal to those of $\gvec(t)$, then the integral ``$\int_0^t d\Gamma(\tau) v(\tau)$" may be interpreted differently using partial sums as
\begin{align} 
\int_0^t d\Gamma(\tau) \ito v(\tau) &:= \lim_{N\to\infty} \sum_{k=0}^{N-1} \tilde \Gamma_k v_k  \label{Eqn: Ito Integrals Definition}\\
\int_0^t d\Gamma(\tau) \strat v(\tau) &:= \lim_{N\to\infty} \sum_{k=0}^{N-1}\tilde \Gamma_k \frac{v_k + v_{k+1}}{2}. \label{Eqn: Stratonovich Integral Definition}
\end{align}
The partial sums are constructed using a partition $\mathcal P_N[0,t]$ as described in Section~\ref{Section: Partition} and by following the notation developed in Section~\ref{Section: Increments} for signals and increments. 
\subsubsection{\textbf{Quadratic Variation}} \label{Section: Definition QV}
The quadratic variation, at time $t$, of a stochastic process $v$ is denoted by $\langle v \rangle (t)$ and is defined using a partition $\mathcal P_N[0,t]$ as
$$ \langle v \rangle (t) := \lim_{N\to \infty} \sum_{k=0}^{N-1} \vnorm{\tilde v_k}^2.$$
\subsubsection{\textbf{Hadamard Product and the Diagonal Operator}}
For any two matrices $A$ and $B$ of the same dimensions, their Hadamard (or element-by-element) product is denoted by ${A \circ B}$. 
For any vector $v$ (resp. square matrix $V$), $\mathcal D(v)$ (resp. $\mathcal D(V)$) denotes a diagonal matrix whose diagonal elements are equal to $v$ (resp. diagonal entries of $V$). 

\section{Problem Formulation} \label{Section: Problem}
In this section, we first provide a precise definition for Mean-Square Stability (MSS) from a purely input/output approach. Then we present a ``stochastic block diagram" formalism that can be given a desirable interpretation by prescribing a suitable stochastic calculus (It\=o or Stratonovich). 
\subsection{Input-Output Formulation of MSS}
Let $\mathcal M$ be a causal LTI (MIMO) system. It is defined as a linear operator that acts on the differential of a second order stochastic signal $u$, denoted by $du$. Its action is defined by the stochastic convolution integral
\begin{equation} \label{Eqn: Stochastic Convolution}
y(t) = \big(\mathcal M du\big)(t) \Longleftrightarrow y(t) = \int_{0}^t M(t-\tau) ~ du(\tau),
\end{equation}
where $M$ is a deterministic matrix-valued function denoting the impulse response of $\mathcal M$. Without loss of generality,  zero initial conditions are assumed throughout this paper. When $u$ is zero-mean and has independent increments such that $\expec{du(t) du^*(\tau)} = 0 ~\forall t\neq \tau$ and $\expec{du(t)du^*(t)} = \cov U(t) dt$, a standard calculation relates the input and output instantaneous covariances as 
\begin{align}
	\cov{Y}(t)  =    \int_{0}^t    M(t-\tau)  ~\cov{U}(\tau)~ M^*(t-\tau) d\tau.	\label{Eqn: IO Cov}	
\end{align}
Note that (\ref{Eqn: IO Cov}) holds for any stochastic interpretation (eg. It\=o or Stratonovich) of the stochastic integral in (\ref{Eqn: Stochastic Convolution}) as shown in Appendix~\ref{Section: Stochastic Convolution}. Therefore, the action of $\mathcal M$ as described in (\ref{Eqn: Stochastic Convolution}) is not given a particular stochastic interpretation throughout the paper. 
Unlike (\ref{Eqn: Stochastic Convolution}), this matrix convolution relationship is deterministic, and it is only valid when the input $du$ is temporally independent (i.e. $u$ has independent increments). Taking the trace of both sides of (\ref{Eqn: IO Cov}) yields
\begin{align*}
	\cov y(t) &= \tr{\cov Y(t)} = \int_0^t \text{tr}\big(M(t-\tau) \cov U(\tau) M^*(t-\tau)\big) d\tau \\
	&= \int_0^t \text{tr} \big(M^*(t-\tau) M(t-\tau) \cov U(\tau)  \big) d\tau \\
	& \leq \int_0^t \text{tr} \big(M^*(t-\tau) M(t-\tau) \big) \text{tr} \big( \cov U(\tau) \big)d\tau \\
	& \leq \int_0^\infty \text{tr} \big(M^*(t-\tau) M(t-\tau)\big) d\tau \sup_{0\leq \tau \leq \infty} \cov u(\tau),
\end{align*}
where the first inequality holds because for any two positive semidefinite matrices $A$ and $B$, we have ${\tr{AB} \leq \tr{A} \tr{B}}$ \cite[Thm 1]{coope1994matrix}. 
The calculation above motivates the following definition for input/output MSS when the input is temporally independent.
\begin{definition} \label{Def: IO MSS}
	A causal LTI system $\mathcal M$ is \textit{Mean-Square Stable} (MSS) if for each input $du$, representing  the differential of a stochastic process with independent increments and uniformly bounded variance, the  output process $y=\mathcal M du$ has a uniformly bounded variance, i.e. there exists a constant $c$ such that $\cov y(t) \leq c~\sup_\tau \cov u(\tau) $.
\end{definition}
	
It is easy to check that $\mathcal M$ is MSS in the sense of Definition~\ref{Def: IO MSS} if and only if $\|\mathcal M\|_2$ is finite, where $\vnorm{.}_2$ denotes the $H^2-\text{norm}$. When MSS holds, the output covariance has a finite steady-state limit $\sscov Y$ whenever the input covariance has a finite steady-state limit $\sscov U$. From (\ref{Eqn: IO Cov}), it is straight forward to see that the steady-state covariances (if they exist) are related as
\begin{align} \label{Eqn: SS IO COv}
	\sscov Y = \int_0^\infty M(\tau) \sscov U M^*(\tau) d\tau.
\end{align}

\subsection{Stochastic Feedback Interconnection}
Consider the ``stochastic block diagram" depicted in Figure~\ref{Fig: MSS Setting} where the forward block represents a causal LTI system which is in feedback with multiplicative stochastic gains represented here as the differential of a diagonal matrix denoted by $d \Gamma(t)$ where
\begin{align} \label{Eqn: dGamma}
	d\Gamma(t) :=\mathcal D\big(d\gvec(t)\big)
	\quad \text{and} \quad  d\gvec(t) := \begin{bmatrix} d\gamma_1(t) & \cdots & d\gamma_n(t) \end{bmatrix}^*.
\end{align}  
Furthermore, a different type of stochastic disturbance enters the dynamics additively and is represented in Figure~\ref{Fig: MSS Setting} as the differential of $w$.
\begin{figure}[h!]
	\centering
	\includegraphics[scale=0.58]{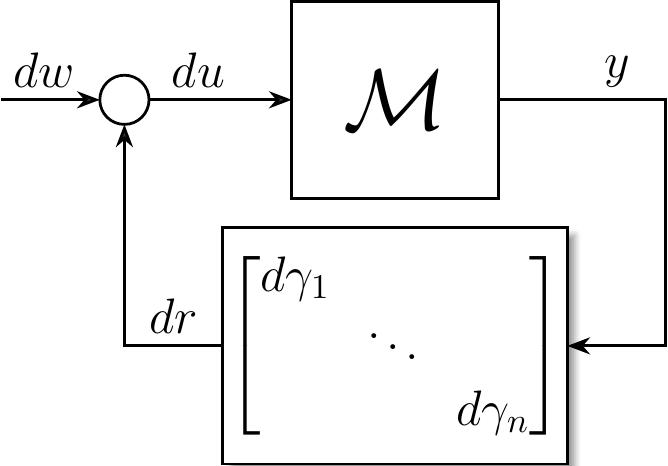}
	\caption{\footnotesize{A continuous-time setting for a causal LTI system $\mathcal M$ in feedback with stochastic multiplicative gains $\{d\gamma_i\}$ that represent the differential forms of, possibly mutually correlated, Wiener processes. The equations describing the block diagram are given in (\ref{Eqn: Equations of MSS Setting}).}}
	\label{Fig: MSS Setting}
\end{figure}

The main objective of this paper is to investigate the MSS of Figure~\ref{Fig: MSS Setting} under the following assumptions
\begin{itemize}
	\item
	\begin{assumption} \label{Ass: M}
		$\mathcal M$ is a causal LTI (MIMO) system whose impulse response $M$ belongs to the class $\mathcal C$ of deterministic, matrix-valued functions defined in Appendix~\ref{Section: Variations of M}. Note that for such $M$, $\exists$ a continuous scalar function $c_M$ such that $\sup\limits_{0\leq \tau \leq t} \mnorm{M(\tau)} = c_M(t).$
	\end{assumption}
	\item 
	\begin{assumption} \label{Ass: gamma}
		$\gvec(t) := \begin{bmatrix} \gamma_1(t) & \cdots & \gamma_n(t) \end{bmatrix}^*$ is a zero-mean, vector-valued Wiener process with an instantaneous covariance $\expec{\gvec(t) \gvec^*(t)}:= \cov \Gamma t$ which can be equivalently written as $\expec{d\gvec(t) d\gvec^*(t)} = \cov \Gamma dt$ (refer to Section~\ref{Section: Wiener Process}). Note that $\cov \Gamma$ is a constant positive semidefinite matrix.
	\end{assumption}
	\item 
	\begin{assumption} \label{Ass: w}
		$w$ is a zero-mean, vector-valued Wiener process with a (possibly) time-varying instantaneous covariance matrix, i.e. $\expec{dw(t)dw^*(t)} = \cov W(t)dt$, where $\cov W$ is a positive semidefinite matrix whose entries remain bounded for all time. Furthermore, $\cov W$ is assumed to be monotone, i.e. if $t_1 \leq t_2$ then $\cov W(t_1) \leq \cov W(t_2)$.
	\end{assumption}
	\item 
	\begin{assumption} \label{Ass: Uncorrelated}
		$\gvec$ and $w$ are uncorrelated for all time.
	\end{assumption}
\end{itemize}
Throughout the paper, whenever the Stratonovich interpretation is adopted, a more restrictive assumption on $M$ is required for reasons that will become apparent in Section~\ref{Section: MS Equivalence}. Thus Assumption~\ref{Ass: M} is replaced by
\begin{itemize}
	\item 
	\begin{customass}{1$^\prime$} \label{Ass: M Lipschitz}
		$M$ is Lipschitz continuous. 
	\end{customass}
\end{itemize}
Note that the class of Lipschitz continuous functions is more restrictive than class $\mathcal C$ defined in Appendix~\ref{Section: Variations of M}. In fact, it is fairly straightforward to see that if $M$ is Lipschitz continuous, then $M \in \mathcal C$.

The equations describing the block diagram in Figure~\ref{Fig: MSS Setting} can be written as
\begin{align} \label{Eqn: Equations of MSS Setting}
	\left\{
	\begin{aligned} 
		y(t) &= \left(\mathcal M du\right)(t)\\
		du(t) &= dw(t) + dr(t)\\
		dr(t) &= d\Gamma(t) y(t) .
	\end{aligned}  \right. 
\end{align}
Note that, without prescribing a stochastic interpretation for the calculus operations on the Wiener processes $w$ and $\Gamma$, the set of equations in (\ref{Eqn: Equations of MSS Setting}) are not sufficient to fully describe the underlying stochastic dynamics. In this paper, we consider the two most common interpretations named after It\=o and Stratonovich; however, the analysis can be generalized to other interpretations as well. We encode the stochastic interpretations in (\ref{Eqn: Equations of MSS Setting}) by rewriting them as
\begin{align} \label{Eqn: Equations of MSS Setting with Interp}
	\left\{
	\begin{aligned} 
		y(t) &= \left(\mathcal M du\right)(t)\\
		du(t) &= dw(t) + dr(t)\\
		dr(t) &= d\Gamma(t) \diamond y(t); \qquad \qquad \text{for} \quad \diamond = \{\ito, \strat\},
	\end{aligned}  \right.
\end{align}
where the last equation is the differential form of an integral equation that can be written as 
$$r(t) = \int_0^t d\Gamma(\tau) \diamond y(\tau), \qquad \text{where} \quad \diamond = \{\ito,\strat\}.$$
Refer to Section~\ref{Section: Interpretations} for an explanation of the different interpretations. Note that We close this section by giving a definition for MSS of the stochastic feedback system in Figure~\ref{Fig: MSS Setting} by following the convention given in \cite{desoer1975feedback}.
\begin{definition}
	Consider the stochastic feedback interconnection in Figure~\ref{Fig: MSS Setting} satisfying Assumptions \ref{Ass: M}-\ref{Ass: Uncorrelated}. The overall feedback system is said to be MSS if all the signals in the loop, i.e. $du, dr$ and $y$ have uniformly bounded variances. More precisely, there exists a constant $c$ such that
	$$ \text{max} \{\vnorm{\cov u}_\infty, \vnorm{\cov r}_\infty, \vnorm{\cov y}_\infty\} \leq c \vnorm{\cov w}_\infty.$$
\end{definition}
The next section characterizes the conditions of MSS for Figure~\ref{Fig: MSS Setting} for different stochastic interpretations. 

\section{Main Results} \label{Section: Results}
Observe that the set of equations (\ref{Eqn: Equations of MSS Setting with Interp}) can be rewritten as a single equation
\begin{equation} \label{Eqn: Single SIE}
	\begin{aligned}
		y(t) &= \int_0^t M(t-\tau) dw(\tau) + \int_0^t M(t-\tau) \diamond d\Gamma(\tau)  y(\tau); \\ & \qquad \qquad \qquad \qquad \qquad 
		\text{for}~\diamond = \{\ito, \strat\}.
	\end{aligned}
\end{equation}
Equation (\ref{Eqn: Single SIE}) is a linear Stochastic Integral Equation (SIE) of Volterra type. The It\=o version of (\ref{Eqn: Single SIE}) has been addressed in the literature (\cite{ito1979existence}, \cite{berger1980volterra1}, \cite{berger1980volterra2}, \cite{berger1979theorems}). For example, it is easy to check that (\ref{Eqn: Single SIE}), interpreted in the sense of It\=o, has a unique solution \cite[Thm 5A]{berger1979theorems} under the assumption that $M$ is finite over bounded intervals (Assumption~\ref{Ass: M}). However, SIEs interpreted in the sense of Stratonovich are less common in the literature. In contrast, SDEs interpreted in the sense of Stratonovich \cite{willems1973mean} are analyzed by converting them to their equivalent It\=o representation using the conversion formulas that were derived several decades ago (see e.g. \cite{stratonovich1966new}). In the present paper, the analysis is carried out from a purely input-output approach, and thus a more general conversion formula is required to convert an SIE interpreted in the sense of Stratonovich to its equivalent It\=o counterpart. In this section, we first describe the conversion scheme, then state the MSS conditions of Figure~\ref{Fig: MSS Setting} when different stochastic interpretations are adopted. 

\subsection{Block Diagram Conversion from Stratonovich to It\=o Interpretations} \label{Section: Ito Stratonovich Conversion}
Consider the block diagram in Figure~\ref{Fig: MSS Setting Stratonovich}(a) such that Assumptions~\ref{Ass: M Lipschitz}, \ref{Ass: gamma}, \ref{Ass: w}, and \ref{Ass: Uncorrelated} are satisfied. As opposed to Figure~\ref{Fig: MSS Setting}, the multiplicative gains are now given a Stratonovich interpretation indicated by the symbol ``$\strat$" in the feedback block. 
Now we present a theorem that describes a conversion scheme of block diagrams from  Stratonovich to It\=o interpretations.
\begin{theorem} \label{Thm: Mean Square Equivalence}
	\textit{Under Assumptions~\ref{Ass: M Lipschitz}, \ref{Ass: gamma}, \ref{Ass: w}, and \ref{Ass: Uncorrelated}, the two block diagrams in Figures~\ref{Fig: MSS Setting Stratonovich}(a) and (b) are equivalent in the mean-square sense. That is, all the signals $du$, $y$, $dw$ and $dr$ in both block diagrams are equal in the mean-square sense.}
\end{theorem}
\begin{figure}[h!]
	\centering
	\begin{subfigure}{0.24\textwidth}
		\includegraphics[scale=0.6]{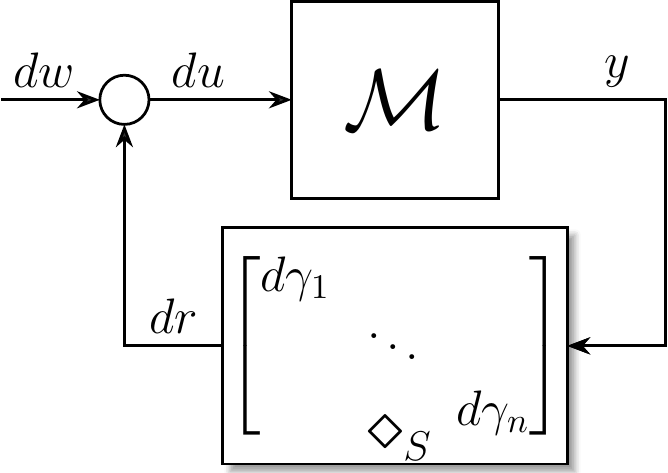} 
		\caption{\footnotesize{Stratonovich Interpretation}}
	\end{subfigure}
	\begin{subfigure}{0.24\textwidth}
		\includegraphics[scale=0.5]{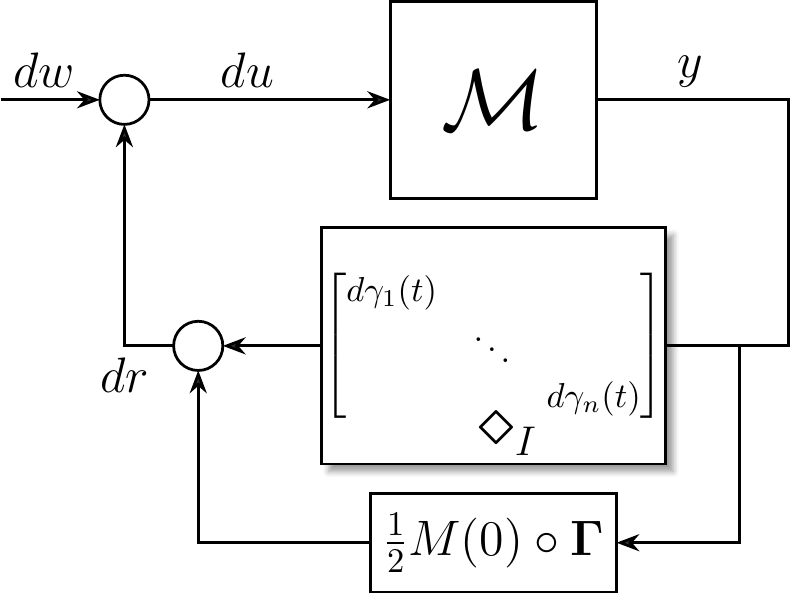} 
		\caption{\footnotesize{Equivalent It\=o Interpretation}}
	\end{subfigure}
	\caption{\footnotesize{(a) A continuous-time causal LTI system $\mathcal M$ in feedback with stochastic multiplicative gains $\{d\gamma_i\}$ that represent the differential forms of, possibly mutually correlated, Wiener processes. The diamond "$\strat$" in the feedback block indicates a Stratonovich interpretation. (b) The equivalent It\=o interpretation, in the mean-square sense, of the block diagram given in (a). The symbol ``$\circ$" denotes the Hadamard (element-by-element) product and ``$\ito$" indicates an It\=o interpretation of the multiplicative gains.}}
	\label{Fig: MSS Setting Stratonovich}
\end{figure}
The proof of Theorem~\ref{Thm: Mean Square Equivalence} is given in Section~\ref{Section: MS Equivalence}. A remark is worth noting here.
\begin{remark} \label{Remark: Relative Degree}
	If $M(0) = 0$, the block diagrams in Figures~\ref{Fig: MSS Setting Stratonovich} (a) and (b) become identical. This means that there is no difference between It\=o and Stratonovich interpretations if the impulse response is zero at initial time. This sort of reintroduces a notion of ``strict causality" that forces the Stratonovich interpretation to behave in the same way as that of It\=o. Therefore, LTI systems $\mathcal M$ with relative degrees \footnote{The relative degree of an LTI system with impulse response $M$ is defined as the largest positive integer $p$ such that $\lim_{s\to\infty} s^pM(s) < \infty$.} $\geq 2$ have the same MSS conditions for both It\=o and Stratonovich interpretations.
\end{remark}

\subsection{Mean-Square Stability Conditions} \label{Section: MSS Conditions}
The MSS setting considered in this paper is given in Figure~\ref{Fig: MSS Setting} and is repeated here in Figure~\ref{Fig: MSS Setting Two Interpretations} to explicitly show the adopted stochastic interpretation of the feedback block. In this section, MSS conditions are given in terms of a linear operator, denoted by $\mathbb L$, that acts on a positive semidefinite matrix to produce another positive semidefinite matrix.
\begin{figure}[h!]
	\centering
	\includegraphics[scale=0.6]{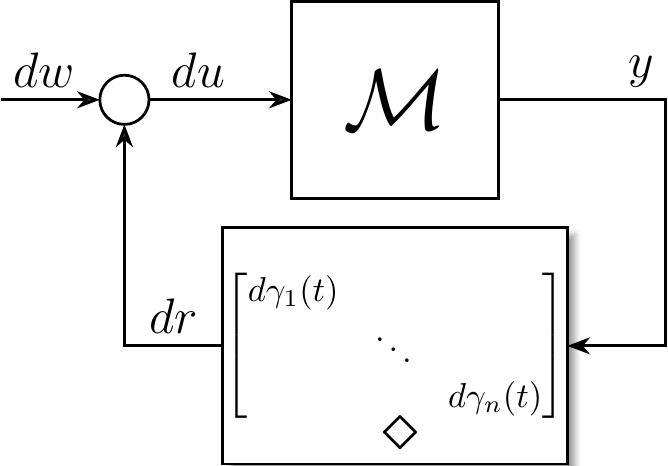} 
	\caption{\footnotesize{Mean-square stability setting. This figure is similar to the general setting given Figure~\ref{Fig: MSS Setting}. The only difference is that the stochastic interpretation of the feedback block is encoded by the symbol ``$\diamond$" such that $\diamond = \ito$ denotes an It\=o interpretation, whereas $\diamond = \strat$ denotes a Stratonovich interpretation. }}
	\label{Fig: MSS Setting Two Interpretations}
\end{figure}
Its role is to propagate the steady-state covariance (if it exists) of $du$, denoted by $\sscov U$, through the loop to yield that of $dr$, denoted by $\sscov R$. This ``Loop Gain Operator" (LGO) is the continuous-time counterpart of that defined in \cite{bamieh2018structured} for the discrete-time setting. For the It\=o setting (i.e. $\diamond = \ito$ in Figure~\ref{Fig: MSS Setting Two Interpretations}), the LGO is denoted by $\mathbb L_I$ and is given by
\begin{align} \label{Eqn: Ito LGO}
	\sscov R = \mathbb L_I\left(\sscov U\right) &:= \cov \Gamma \circ \left(\int_0^\infty M(\tau) \sscov U M^*(\tau) d\tau\right).
\end{align}
Refer to Section~\ref{Section: LGO and MSS} for a detailed derivation of the LGO. A key step in the derivation of $\mathbb L_I$ is showing that $du$ is temporally independent which is required to propagate $\sscov U$ in the forward block $\mathcal M$ using (\ref{Eqn: SS IO COv}). As will be shown in Section~\ref{Section: Stochastic Block Diagram Interpretation}, this temporal independence is a consequence of (1) the causality of $\mathcal M$, (2) the temporal independence of the stochastic multiplicative gains, and (3) the It\=o interpretation. However, for the Stratonovich setting (i.e. $\diamond = \strat$ in Figure~\ref{Fig: MSS Setting Two Interpretations}), $du$ is not temporally independent. This is a consequence of the nature of the Stratonovich integral in (\ref{Eqn: Stratonovich Integral Definition}) that ``looks into the future". In this case, (\ref{Eqn: SS IO COv}) cannot be used to propagate the covariance in the forward block of Figure~\ref{Fig: MSS Setting Stratonovich}(a). Nonetheless, one can exploit the block diagram conversion scheme in Section~\ref{Section: Ito Stratonovich Conversion} and rearrange the block diagram in Figure~\ref{Fig: MSS Setting Stratonovich}(b) so that it looks like the It\=o setting as depicted in Figure~\ref{Fig: Rearranged Ito Equivalent}. 
The equivalent forward block, now denoted by $\mathcal H$, is still a causal LTI system whose transfer function is 
\begin{align} \label{Eqn: Transfer Functions}
	H(s) = \left(I - M(s) G \right) ^{-1} M(s),
\end{align}
where $G := \frac{1}{2} M(0) \circ \cov \Gamma$ and $M(s)$ is the transfer function of $\mathcal M$. 
\begin{figure}[h!]
	\centering
	\includegraphics[scale=0.6]{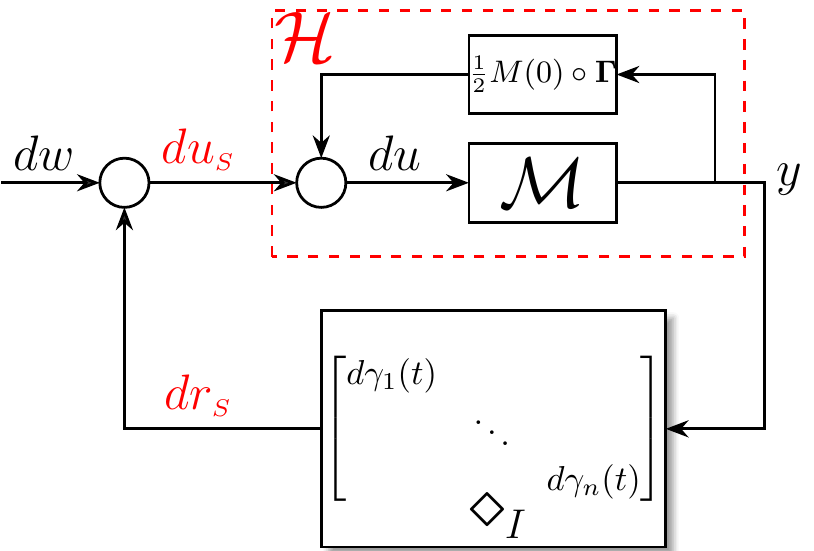} 
	\caption{\footnotesize{Rearrangement of the block diagram in Figure~\ref{Fig: MSS Setting Stratonovich}(b) }}
	\label{Fig: Rearranged Ito Equivalent}
\end{figure}
The input differential signal $du_\mathsmaller{S}$ in Figure~\ref{Fig: Rearranged Ito Equivalent} is now temporally independent and thus (\ref{Eqn: SS IO COv}) can be exploited to propagate the steady state covariance through the equivalent forward block $\mathcal H$.
Thus, the LGO for the Stratonovich setting propagates the steady-state covariance (if it exists) of $du_{\mathsmaller S}$, denoted by $\sscov U_{\mathsmaller S}$, through the loop of Figure~\ref{Fig: Rearranged Ito Equivalent} to yield that of $dr_{\mathsmaller S}$, denoted by $\sscov R_{\mathsmaller S}$. It is now denoted by $\mathbb L_S$ and is given by
\begin{align} \label{Eqn: Stratonovich LGO}
	\sscov R_{\mathsmaller S} = \mathbb L_S\left(\sscov U_{\mathsmaller S}\right) &:= \cov \Gamma \circ \left(\int_0^\infty H(\tau) \sscov U_{\mathsmaller S} H^*(\tau) d\tau\right),
\end{align}
where $H$ is given in (\ref{Eqn: Transfer Functions}).
The spectral radius of $\mathbb L$ completely characterizes the MSS condition as will be seen next.
\begin{theorem} \label{Thm: MSS Conditions}
\textit{Consider the system in Figure~\ref{Fig: MSS Setting Two Interpretations} such that Assumptions~\ref{Ass: M}-\ref{Ass: Uncorrelated} are satisfied. The feedback system is MSS if and only if the two conditions are satisfied
\begin{enumerate}
	\item The equivalent forward block in Figure~\ref{Fig: MSS Setting Two Interpretations} has a finite $H^2-\text{norm}$.
	\item The spectral radius of the loop gain operator is strictly less than 1, i.e. $\rho(\mathbb L) <1$.
\end{enumerate}
where 
\begin{itemize}
	\item For the It\=o interpretation, the equivalent forward block is $\mathcal M$, and $\mathbb L$ is given in (\ref{Eqn: Ito LGO}).
	\item For the Stratonovich interpretation, the equivalent forward block is $\mathcal H$, whose transfer function is given in (\ref{Eqn: Transfer Functions}), $\mathbb L$ is given in (\ref{Eqn: Stratonovich LGO}), and Assumption~\ref{Ass: M} is replaced by Assumption~\ref{Ass: M Lipschitz}.
\end{itemize}}
\end{theorem}
The proof of Theorem~\ref{Thm: MSS Conditions} is given in Section~\ref{Section: LGO and MSS}.
Observe that, under the It\=o interpretations, the covariance matrix $\cov \Gamma$ only plays a role in the second condition. However, under the Stratonovich interpretation, $\cov \Gamma$ plays a role in both conditions since the equivalent forward block $\mathcal H$ now depends on $\cov \Gamma$ (Figure~\ref{Fig: Rearranged Ito Equivalent}). Therefore, the conditions of MSS can be very different when different stochastic interpretations are adopted. We close this section by noting that the spectral radius of $\mathbb L$ can be numerically calculated using the power iteration explained in \cite{bamieh2018structured}. 

\section{Application to State Space Realizations \& SDEs} \label{Section: SS}
In this section, we consider the mean-square stability problems for both the It\=o and Stratonovich settings given in Figure~\ref{Fig: MSS Setting Two Interpretations}, but for the special case when $\mathcal M$ is given a state space realization. Thus, the underlying equations can be written as SDEs, i.e.
\begin{align} \label{Eqn: SS Realization}
dx(t) &= Ax(t) dt + Bdu(t); \qquad
y(t) = Cx(t) \nonumber\\
du(t) &= dw(t) + dr(t) \nonumber \\
dr(t) &= d\Gamma(t) \diamond y(t) \quad\text{for} \quad \diamond = \{\ito, \strat\},
\end{align}
where the last equation refers to either an It\=o or Stratonovich interpretation. The impulse response of $\mathcal M$ can thus be written as $M(t) = C e^{At} B$. Then, the realization of the loop gain operator, for each interpretation, can be calculated using (\ref{Eqn: Ito LGO}) and (\ref{Eqn: Stratonovich LGO}). Starting with the It\=o interpretation, we have
\begin{align*}
\sscov R &= \mathbb L_I(\sscov U) := \cov \Gamma \circ \left(\int_0^\infty M(\tau) \sscov U M^*(\tau) d\tau \right)\\
&= \cov \Gamma \circ \left(C \int_0^\infty e^{A\tau} B \sscov U B^* e^{A^* \tau} d\tau \right) C \\
&= \cov \Gamma \circ \left (C \sscov X C\right), 
\end{align*}
where $\sscov X := \int_0^\infty e^{A\tau} B\sscov U B^* e^{A^*\tau} d\tau$ which satisfies the algebraic Lyapunov equation given by
$$ A \sscov X + \sscov X A^* + B \sscov U B^* = 0.$$

For the Stratonovich interpretation, we use Figure~\ref{Fig: Rearranged Ito Equivalent} to give the equivalent It\=o representation. The impulse response of $\mathcal H$ in Figure~\ref{Fig: MSS Setting Stratonovich}(b) can be  shown to be $H(t)=Ce^{A_S t}$ with $A_S = A + 1/2 B \big((CB) \circ \cov \Gamma \big) C$ and the LGO can be similarly given a realization. To summarize, let $\mathbb L_I$ and $\mathbb L_S$ denote the loop gain operators for the It\=o and Stratonovich interpretations as given in (\ref{Eqn: Ito LGO}) and (\ref{Eqn: Stratonovich LGO}), respectively. Then their state space realizations are given by
\begin{align} \label{Eqn: SS LGO}
&\begin{aligned}
\sscov R &= \mathbb L_k(\sscov U) \\
(k &= I, S)
\end{aligned}
& \Leftrightarrow &
&\left\{ 
\begin{aligned} 
\sscov R &= \cov \Gamma \circ \left(C\sscov X C^*\right) \\
0 &= A_k \sscov X + \sscov X A_k^* + B \sscov U B^*; 
\end{aligned} \right. 
\end{align}
where $A_I := A$ and $A_S:= A + \frac{1}{2} B\big((CB) \circ \cov \Gamma\big) C$.
Therefore, as a direct application of Theorem~\ref{Thm: MSS Conditions}, the necessary and sufficient conditions of MSS are (1) $A_k$ is Hurwitz and (2) $\rho(\mathbb L_k) < 1$ for $k = I,S$ for It\=o and Stratonovich interpretations, respectively.

\section{Stochastic Block Diagram Conversion Technique} \label{Section: MS Equivalence}
In this section, we provide a proof for Theorem~\ref{Thm: Mean Square Equivalence}. Consider the Stratonovich setting in Figure~\ref{Fig: MSS Setting Stratonovich}(a) such that Assumptions~\ref{Ass: M Lipschitz}, \ref{Ass: gamma}, \ref{Ass: w}, and \ref{Ass: Uncorrelated} are satisfied. The block diagram can be described by a single SIE given in (\ref{Eqn: Single SIE}) with $\diamond = \strat$, and the goal of this section is to show that it is equivalent (in the mean-square sense) to
\begin{align} \label{Eqn: Single SIE Ito}
	y(t) &= \int_0^t M(t-\tau) dw(\tau) + \int_0^t M(t-\tau) \ito d\Gamma(\tau) y(\tau) \nonumber \\
	&\qquad \qquad \qquad + \frac{1}{2} \int_0^t M(t-\tau) \big(M_0 \circ \cov \Gamma \big) y(\tau) d\tau,
\end{align}
where $M(0)$ is denoted by $M_0$ for notational convenience. This can be shown by exploiting the following two propositions.
\begin{proposition} \label{Proposition: Bounded Moments}
	\textit{Consider the SIE given in (\ref{Eqn: Single SIE Ito}) (or equivalently (\ref{Eqn: Single SIE}) with $\diamond = \strat$) such that Assumptions~\ref{Ass: M Lipschitz}, \ref{Ass: gamma}, \ref{Ass: w}, and \ref{Ass: Uncorrelated} are satisfied. Then the second moments of $y$ and its quadratic variation (Section~\ref{Section: Definition QV}) are both finite over finite intervals. That is, there exist two scalar continuous functions $c_y$ and $c_q$ such that
	\begin{equation} \label{Eqn: Bounded Moments}
		\sup_{0\leq \tau \leq t} \expec{\vnorm{y(\tau)}^2} = c_y(t); 
		\sup_{0\leq \tau \leq t} \expec{\langle y \rangle^2(\tau)} = c_q(t).
	\end{equation} }
	The proof of the boundedness of $\expec{\vnorm{y(\tau)}^2}$ is given in \cite[Thm 5A]{berger1979theorems} while that of the quadratic variation is given in Section~\ref{Section: Bounded Quadratic Variation}. These bounds will be useful to prove Proposition~\ref{Proposition: Strato to Ito}.
\end{proposition}
\begin{proposition} \label{Proposition: Strato to Ito}
	\textit{Consider the  Stratonovich integral 
	$$S(t) := \int_0^t M(t-\tau) d\Gamma(\tau) \strat  y(\tau), $$
	where $M$ satisfies Assumption~\ref{Ass: M}, $d\Gamma(t)$ is defined in (\ref{Eqn: dGamma}) such that $\gvec$ satisfies Assumption~\ref{Ass: gamma}, and $y$ is a stochastic process that satisfies (\ref{Eqn: Single SIE}) with $\diamond = \strat$. Then $S(t) = I(t) + \frac{1}{2} R(t)$ in the mean-square sense, where
	\begin{equation*}
		\begin{aligned}
			I(t) &:= \int_0^t M(t-\tau)  \ito d\Gamma(\tau) y(\tau) \qquad \text{and}\\
			R(t) &:=  \int_0^t M(t-\tau)\big(M_0 \circ \cov \Gamma \big) y(\tau) d\tau
		\end{aligned}
	\end{equation*}
	are It\=o and Riemann integrals, respectively.}
\end{proposition}
\begin{proof} 
	Start by using the definitions of the various integrals in Section~\ref{Section: Interpretations} to construct the partial sums over a partition $\mathcal P_N[0,t]$ (\ref{Section: Partition}) as
	\begin{equation} \label{Eqn: Integrals Definitions}
		\begin{aligned}
			S_N(t) &:= \frac{1}{2} \sum_{k=0}^{N-1} \Big(M(t-t_{k+1}) \tilde \Gamma_k y_{k+1} + M(t-t_k) \tilde \Gamma_k y_k \Big) \\
			I_N(t) &:= \sum_{k=0}^{N-1} M(t-t_k) \tilde \Gamma_k y_k\\
			R_N(t) &:=  \sum_{k=0}^{N-1} M(t-t_k) \big(M_0 \circ \cov \Gamma \big)y_k \Delta_k.
		\end{aligned}
	\end{equation}
	The proof is carried out on the partition $\mathcal P_N[0,t]$ but can be passed to the limit in $L_2(p)$ (since it is a Hilbert space and all Cauchy sequences are convergent). More precisely, we are required to prove that $\lim_{N\to\infty} \expec{D_N^2(t)} = 0 ~ \forall t \geq 0,$
	\begin{equation} \label{Eqn: Def DN}
		\text{where} \qquad D_N(t) = S_N(t) - \left(I_N(t) + \frac{1}{2} R_N(t)\right).
	\end{equation} 
	After carrying out a sequence of algebraic manipulations (Appendix~\ref{Section: Algebraic Manipulations}), the expression of $D_N(t)$ can be rewritten as
	\begin{equation} \label{Eqn: DN}
		\begin{aligned}
			D_N(t) &= \frac{1}{2} \Big( \lambda_N(t) + J_N(t) + \nu_N(t) + \xi_N(t) + T^\zeta_N(t)\Big) \\
			& \qquad+ \frac{1}{4} \Big(\theta_N(t) + \eta_N(t) + T^\alpha_N(t) + T^\beta_N(t) \Big),
	\end{aligned}
	\end{equation}
	where
	\begin{equation} \label{Eqn: Sums}
		\begin{aligned}
			\lambda_N(t) &:= \sum_{k=0}^{N-1}M(t-t_k) \Big( \big(\tilde{\boldsymbol \gamma}_k \tilde{\boldsymbol \gamma}_k^* - \cov \Gamma \Delta_k\big) \circ M_0 \Big) y_k \\
			J_N(t) &:= \sum_{k=0}^{N-1} \bigg( M(t-t_{k+1}) - M(t - t_k) \bigg) \tilde \Gamma_k y_k \\
			\nu_N(t) &:= \sum_{k=0}^{N-1} \Big( M(t-t_{k+1}) - M(t-t_k)\Big) \tilde \Gamma_k M_0 \tilde \Gamma_k y_k \\
			\theta_N(t) &:= \sum_{k=0}^{N-1} M(t-t_{k+1}) \tilde \Gamma_k M_0 \tilde \Gamma_k \tilde y_k \\
			\eta_N(t) &:= \sum_{k=0}^{N-1} M(t - t_{k+1})\tilde \Gamma_k \Big(M(\Delta_k) - M_0 \Big)\tilde \Gamma_k y_k \\
			\chi_N(t) &:= \sum_{k=0}^{N-1} M(t-t_{k+1}) \tilde \Gamma_k M(\Delta_k) \tilde w_k \\
			T_N^x(t) &:= \sum_{k=0}^{N-1} M(t-t_{k+1}) \tilde \Gamma_k x_k \quad \text{for}\quad  x \in\{\alpha, \beta, \zeta\} \\
			\alpha_k &:= \sum_{l=0}^{k-1} \bigg(M(t_{k+1} - t_{l+1}) - M(t_k -t_{l+1}) \bigg) \tilde \Gamma_l \tilde y_l  \\
			\beta_k &:= \sum_{l=0}^{k-1} \bigg( M(t_{k+1} - t_{l+1}) - M(t_k -t_{l+1}) \\ & \qquad \qquad + M(t_{k+1} - t_l) - M(t_k -t_l) \bigg) \tilde \Gamma_l y_l\\
			\zeta_k &:= \sum_{l=0}^{k-1} \bigg(M(t_{k+1} - t_{l}) - M(t_k -t_{l}) \bigg) \tilde w_l.
		\end{aligned}
	\end{equation}
	The rest of the proof shows that the second moment of each term in (\ref{Eqn: DN}) goes to zero in the limit as $N$ goes to infinity. Note that there is no need to check the expectation of cross terms (Appendix~\ref{Section: Cross Terms}).
	\subsubsection{Mean-Square Convergence of $\lambda_N(t)$} 	
	Recall that $\boldsymbol{\gamma}_k$ has independent increments that are also independent from present and past values of $y_k$. Furthermore, $\expec{Z_k} = 0$ with $Z_k:=\tilde{\boldsymbol{\gamma}}_k \tilde{\boldsymbol{\gamma}}_k^* - \cov \Gamma \Delta_k$. Then we invoke Lemma~\ref{Lemma: Ineq Ind zero} to yield the following inequality
	\begingroup\makeatletter\def\f@size{9.2}\check@mathfonts
	\def\maketag@@@#1{\hbox{\m@th\large\mnormalfont#1}}%
	\begin{align*}
		\expec{\vnorm{\lambda_N(t)}^2} & \leq \sum_{k=0}^{N-1} \mnorm{ M(t-t_k)}^2 \expec{\mnorm{ \big(Z_k \circ M_0\big) }^2} \expec{\vnorm{y_k}^2} \\
		& \leq \mnorm{M_0}^2 \sum_{k=0}^{N-1} \mnorm{M(t-t_k)}^2 \expec{ \mnorm{Z_k}^2} \expec{\vnorm{y_k}^2},
	\end{align*} \endgroup
	where the second inequality follows from the sub-multiplicative property of the matrix spectral norm with respect to matrix and Hadamard products (see \cite{horn1990analog}).
	Knowing that $\tilde \gvec_k \sim \mathcal N(0, \cov \Gamma \Delta_k)$, we can write $\tilde{\boldsymbol \gamma}_k = \cov \Gamma^{1/2} \boldsymbol \xi_k \sqrt{\Delta_k}$, where $\cov \Gamma^{1/2}$ denotes the Cholesky factorization of $\cov \Gamma$. The random vector $\boldsymbol \xi_k$ follows a standard multivariate normal distribution for all $k=0,1, ... N-1$ such that $\boldsymbol \xi_k$ and $\boldsymbol \xi_l$ are independent for $k\neq l$. To bound $\expec{\mnorm{Z_k}^2}$, we proceed as follows
	\begin{align*}
		\expec{\mnorm{Z_k}^2 } &= \expec{\mnorm{ \cov \Gamma^{1/2}(\boldsymbol \xi_k \boldsymbol \xi_k^* - I) \cov \Gamma^{1/2}}^2 \Delta_k^2}\\
		&\leq  \expec{\mnorm{\cov \Gamma} \mnorm{\boldsymbol \xi_k \boldsymbol \xi_k^* - I}^2 \Delta_k^2} \\
		&\leq \expec{ \mnorm{\cov \Gamma} \mnorm{\boldsymbol \xi_k \boldsymbol \xi_k^* - I}^2_F \Delta_k^2 }\\
		&=  \expec{ \mnorm{\cov \Gamma} \text{tr}\Big((\boldsymbol \xi_k \boldsymbol \xi_k^* - I)^*(\boldsymbol \xi_k \boldsymbol \xi_k^* - I)\Big) \Delta_k^2} \\
		&= \mnorm{\cov \Gamma} \Delta_k^2 \Big(\expec{\vnorm{\boldsymbol \xi_k}^4} - 2 \expec{\vnorm{\boldsymbol \xi_k}^2} + n \Big) \\
		&= \mnorm{\cov \Gamma} \Delta_k^2 (n^2+n).
	\end{align*}
	where the second inequality follows from the fact that the Frobenius norm of a matrix is larger than its spectral norm. The last equality follows by using Lemma~\ref{Lemma: Moments Normal Vector}, where $n$ is the number of gains $\gamma_i$.
	Finally, we obtain
	\begingroup\makeatletter\def\f@size{9.2}\check@mathfonts
	\def\maketag@@@#1{\hbox{\m@th\large\mnormalfont#1}}%
	\begin{align*} 
		\expec{\vnorm{\lambda_N(t)}^2} & \leq \mnorm{M_0}^2 c_M^2(t)\mnorm{\cov \Gamma} (n^2+n)  c_y(t) \sum_{k=0}^{N-1}  \Delta_k^2 \underset{N\to \infty}{\longrightarrow} 0,
	\end{align*} \endgroup
	where Assumption~\ref{Ass: M} and (\ref{Eqn: Bounded Moments}) are exploited.
	\subsubsection{Mean-Square Convergence of $J_N(t)$}
	This partial sum is similar to that of $\lambda_N(t)$, and thus we define $F_k(t):=M(t-t_{k+1}) - M(t-t_k)$ and invoke Lemma~\ref{Lemma: Ineq Ind zero} again to yield
	\begingroup\makeatletter\def\f@size{9.2}\check@mathfonts
	\def\maketag@@@#1{\hbox{\m@th\large\mnormalfont#1}}%
	\begin{align*}
		\expec{\vnorm{J_N(t)}^2} & \leq \sum_{k=0}^{N-1} \mnorm{F_k(t)}^2 \expec{\mnorm{\tilde \Gamma_k}^2} \expec{\vnorm{y_k}^2} \\
		& \leq c_y(t) \tr{ \cov\Gamma} \sum_{k=0}^{N-1} \mnorm{M(t-t_{k+1}) - M(t-t_k)}^2 \Delta_k \\
		& \leq c_y(t) \tr{\cov \Gamma} \pnorm \qv{t}{M} \underset{N\to\infty}{\longrightarrow} 0,
	\end{align*} \endgroup
	where the second inequality follows from (\ref{Eqn: Bounded Moments}), Lemma~\ref{Lemma: Moments Normal Vector} and the fact that $\mnorm{\tilde \Gamma_k} \leq \vnorm{\tilde \gvec_k}$ since $\tilde \Gamma_k = \mathcal D(\tilde \gvec_k)$ so that
	\begin{align} \label{Eqn: c2}
		\expec{\mnorm{\tilde \Gamma_k}^2} \leq \tr{\cov \Gamma} \Delta_k.
	\end{align}
	The last inequality follows from the fact that the quadratic variation of $M$ is finite (Lemma~\ref{Lemma: Variations}). 
	\subsubsection{Mean-Square Convergence of $\nu_N(t)$}
	By using the same previous definition of $F_k(t)$, invoke Lemma~\ref{Lemma: Ineq Ind} (with $X_k:= \tilde \Gamma_k M_0 \tilde \Gamma_k$) to yield
	\begin{align*}
		&\expec{\vnorm{\nu_N(t)}^2} \\
		& \leq \left( \sum_{k=0}^{N-1} \mnorm{F_k(t)} \sqrtp{\expec{\mnorm{\tilde \Gamma_k M_0 \tilde \Gamma_k}^2} \expec{\vnorm{y_k}^2}} \right)^2 \\
		& \leq c_y(t) \mnorm{M_0}^2 \left( \sum_{k=0}^{N-1} \mnorm{F_k(t)} \sqrtp{ \expec{\mnorm{\tilde \Gamma_k}^4} } \right)^2\\ 
		& ~\Scale[0.83]{ \leq c_y(t) \mnorm{M_0}^2 c(2,n) \mnorm{\cov \Gamma}^2 \left( \sum\limits_{k=0}^{N-1} \mnorm{M(t-t_{k+1}) - M(t-t_k)} \Delta_k \right)^2}\\
		& \leq c_y(t) \mnorm{M_0}^2 c(2,n) \pnorm \Big( \tv{t}{M} \Big)^2\underset{N\to\infty}{\longrightarrow} 0,
	\end{align*}
	where the second inequality follows from (\ref{Eqn: Bounded Moments}) and the sub-multiplicative property of the spectral norm. The third inequality follows from  Lemma~\ref{Lemma: Moments Normal Vector} where
	\begin{align} \label{Eqn: c4}
		\expec{\mnorm{\tilde \Gamma_k}^4} \leq c(2,n) \mnorm{\cov \Gamma}^2 \Delta_k^2,
	\end{align}
	and the last inequality follows from the fact that the total variation of $M$ is finite (Lemma~\ref{Lemma: Variations}). 
	\subsubsection{Mean-Square Convergence of $\eta_N(t)$}
	In a similar fashion to the previous calculation, define $G_k := M(\Delta_k) - M_0$ and invoke Lemma~\ref{Lemma: Ineq Ind} (with $X_k:= \tilde \Gamma_k G_k \tilde \Gamma_k$) to yield
	\begingroup\makeatletter\def\f@size{7.8}\check@mathfonts
	\def\maketag@@@#1{\hbox{\m@th\large\mnormalfont#1}}%
	\begin{align*}
		\expec{\vnorm{\eta_N(t)}^2} & \leq \left( \sum_{k=0}^{N-1} \mnorm{M(t-t_k)} \sqrtp{\expec{\mnorm{\tilde \Gamma_k G_k \tilde \Gamma_k}^2} \expec{\vnorm{y_k}^2}} \right)^2 \\
		& \leq c_y(t) c_M^2(t) \left( \sum_{k=0}^{N-1} \mnorm{M(\Delta_k) - M_0} \sqrtp{ \expec{\mnorm{\tilde \Gamma_k}^4} } \right)^2\\
		& \leq c_y(t) c_M^2(t) c(4,n) \mnorm{\cov \Gamma}^2  \left( \sum_{k=0}^{N-1} \mnorm{M(\Delta_k) - M_0} \Delta_k \right)^2 \\ &\underset{N\to\infty}{\longrightarrow} 0,
	\end{align*}\endgroup
	where the second inequality follows from (\ref{Eqn: Bounded Moments}), Assumption~\ref{Ass: M}, and the sub-multiplicative property of the spectral norm. Again, the last inequality follows from (\ref{Eqn: c4}). The limit is zero because Assumption~\ref{Ass: M} guarantees that $M$ is right-continuous at $t=0$.
	\subsubsection{Mean-Square Convergence of $\chi_N(t)$}
	Since $w$ and $\{\gamma_i\}$ are uncorrelated (Assumption~\ref{Ass: Uncorrelated}), invoking Lemma~\ref{Lemma: Ineq Ind zero} yields
	\begingroup\makeatletter\def\f@size{8}\check@mathfonts
	\def\maketag@@@#1{\hbox{\m@th\large\mnormalfont#1}}%
	\begin{align*}
		\expec{\vnorm{\chi_N(t)}^2} & \leq \sum_{k=0}^{N-1} \mnorm{M(t- t_{k+1})}^2 \expec{\mnorm{\tilde \Gamma_k}^2} \expec{\vnorm{M(\Delta_k) \tilde w_k}^2} \\
		&\leq c_M^4(t) \tr{\cov \Gamma} \sum_{k=0}^{N-1} \Delta_k \tr{\cov W_k} \Delta_k \\
		& \leq c_M^4(t) \tr{\cov \Gamma} c_w(t) \sum_{k=0}^{N-1} \Delta_k^2 \underset{N\to\infty}{\longrightarrow} 0,
	\end{align*}\endgroup
	where the second inequality follows from Assumptions~\ref{Ass: M} and \ref{Ass: w} and (\ref{Eqn: c2}). The last inequality follows because under Assumption~\ref{Ass: w}, $\exists$ a continuous scalar function $c_w$ such that
	\begin{align} \label{Eqn: cw}
		\sup_{0\leq \tau \leq t} \tr{\cov W(\tau)} = c_w(t).
	\end{align}
	\subsubsection{Mean-Square Convergence of $\theta_N(t)$}
	By invoking Lemma~\ref{Lemma: Ineq Dep}, we obtain the following inequality
	\begin{align*}
		\expec{\vnorm{\theta_N(t)}^2} &\leq \sum_{k=0}^{N-1} \mnorm{M(t- t_{k+1})}^2 \sqrtp{\expec{ \mnorm{\tilde \Gamma_k M_0 \tilde \Gamma_k}^4} } \\
		& \qquad \qquad \times \sqrtp{ \expec{\left(\sum_{k=0}^{N-1} \vnorm{ \tilde  y_k}^2\right)^2}},
	\end{align*}
	where the second term converges to $\sqrtp{\expec{\langle y\rangle ^2(t)}} \leq \sqrt{c_q(t)}$ defined in (\ref{Eqn: Bounded Moments}). Now apply the submultiplicative property of the spectral norm to yield
	\begingroup\makeatletter\def\f@size{8}\check@mathfonts
	\def\maketag@@@#1{\hbox{\m@th\large\mnormalfont#1}}%
	\begin{align*}
		&\expec{\vnorm{\theta_N(t)}^2} \leq \sqrt{c_q(t)} \mnorm{M_0}^2 \sum_{k=0}^{N-1} \mnorm{M(t-t_{k+1})}^2 \sqrtp{\expec{\mnorm{\tilde \Gamma_k}^8}} \\
		& \leq \sqrt{c_q(t)} \sqrt{c(4,n)} \mnorm{\cov \Gamma}^2 c_M^2(t) \mnorm{M_0}^2 \sum_{k=0}^{N-1} \Delta_k^2  \underset{N\to\infty}{\longrightarrow} 0,
	\end{align*} \endgroup
	where the last inequality follows from Assumption~\ref{Ass: M} and Lemma~\ref{Lemma: Moments Normal Vector} where $c(4,n) \mnorm{\cov \Gamma}^4 \Delta_k^4$ serves as an upper bound for the eighth moment $\expec{ \vnorm{\tilde \Gamma_k}^8}$. 
	\subsubsection{Mean-Square Convergence of $T_N^\alpha(t), T_N^\beta(t)$ and $T_N^\zeta(t)$}\quad \\
	Observe using (\ref{Eqn: Sums}) that the pairs $(\tilde \Gamma_k, \alpha_k), (\tilde \Gamma_k, \beta_k)$ and $(\tilde \Gamma_k, \zeta_k)$ are independent for all $k = 0, 1, \cdots, N-1$. Then, for $x \in \{\alpha, \beta, \zeta\}$, invoking Lemma~\ref{Lemma: Ineq Ind zero} yields 
	\begin{align*}
		\expec{\vnorm{T_N^x (t)}^2} &\leq \sum_{k=0}^{N-1} \mnorm{M(t-t_{k+1})}^2 \expec{ \mnorm{\tilde \Gamma_k}^2 } \expec{ \vnorm{x_k}^2 } \\
		& \leq c_M^2(t) \tr{\cov \Gamma} \sum_{k=0}^{N-1} \expec{ \vnorm{x_k}^2 } \Delta_k,
	\end{align*}
	where the last inequality follows from Assumption~\ref{Ass: M} and (\ref{Eqn: c2}). 
	Now, we examine $\expec{ \vnorm{\alpha_k}^2 }$. Define $F_{k,l} := M(t_{k+1} - t_{l+1}) - M(t_k - t_{l+1})$ and invoke Lemma~\ref{Lemma: Ineq Dep} to yield
	\begin{align*}
		\expec{ \vnorm{\alpha_k}^2 } &\leq \sum_{l=0}^{k-1} \mnorm{F_{k,l}}^2 \sqrtp{ \expec{\mnorm{\tilde \Gamma_l}^4}} \sqrtp{ \expec{ \sum_{l=0}^{k-1} \vnorm{\tilde y_l}^2 } } \\
		& \leq \sqrt{c(2,n)} \mnorm{\cov \Gamma}\sqrt{c_q(t)} \sum_{l=0}^{k-1} \mnorm{F_{k,l}}^2 \Delta_l \\
		& \leq \sqrt{c(2,n)} \mnorm{\cov \Gamma} \sqrt{c_q(t)} \pnorm \qv{t}{M},
	\end{align*}
	where $\pnorm = \sup_l \Delta_l$. Note that the second inequality follows from (\ref{Eqn: Bounded Moments}) and (\ref{Eqn: c4}), and the third inequality follows by observing that the sum converges to the quadratic variation of $M$ on the interval $[0,t_k]$ (Appendix~\ref{Section: Variations of M}). The last equality exploits the fact that $\qv{t}{M}$ is an increasing function in $t$. Substituting in $\expec{\vnorm{T_N^\alpha (t)}^2}$ yields
	\begingroup\makeatletter\def\f@size{7.8}\check@mathfonts
	\def\maketag@@@#1{\hbox{\m@th\large\mnormalfont#1}}%
	\begin{align*}
		\expec{\vnorm{T_N^\alpha (t)}^2} &\leq c_M^2(t) \tr{\cov \Gamma} \sqrt{c(2,n)} \mnorm{\cov \Gamma} \sqrt{c_q(t)} \pnorm \qv{t}{M} \sum_{k=0}^{N-1} \Delta_k\\
		& \leq c_M^2(t) \tr{\cov \Gamma} \sqrt{c(2,n)} \mnorm{\cov \Gamma} \sqrt{c_q(t)} \pnorm \qv{t}{M} t \underset{N\to\infty}{\longrightarrow} 0.
	\end{align*} \endgroup
	Recalling from Appendix~\ref{Section: Cross Terms} that there is no need to check the convergence of the cross terms, the same arguments used for $\expec{\vnorm{T_N^\alpha (t)}^2}$ can be used here to show that 
	\begin{align*}
		\expec{\vnorm{T_N^\beta (t)}^2} \underset{N\to\infty}{\longrightarrow} 0 \qquad \text{and} \qquad \expec{\vnorm{T_N^\zeta (t)}^2} \underset{N\to\infty}{\longrightarrow} 0.
	\end{align*}
	This completes the proof of Proposition~\ref{Proposition: Strato to Ito}.
\end{proof}
A direct application of Proposition~\ref{Proposition: Strato to Ito} to (\ref{Eqn: Single SIE}) with $\diamond = \strat$ yields (\ref{Eqn: Single SIE Ito}). This is exactly the result shown in Figure~\ref{Fig: MSS Setting Stratonovich}(b) and given in Theorem~\ref{Thm: Mean Square Equivalence}.

\section{Loop Gain Operator \& MSS Conditions} \label{Section: LGO and MSS}
In this section, we give the mathematical derivations of the LGO (\ref{Eqn: Ito LGO}) for the It\=o setting. The same analysis can be carried out for the Stratonovich case by using the conversion scheme developed in Section~\ref{Section: Ito Stratonovich Conversion}. We first lay down the necessary framework to construct a deterministic block diagram that describes the continuous-time evolution of the covariance matrices of the various signals in the loop (see Figure~\ref{Fig: Covariance Diagram}). Once this deterministic setting is constructed, the MSS analysis from there onwards resembles that of the discrete-time counterpart in \cite{bamieh2018structured}. 
\subsection{Stochastic Block Diagram Interpretation} \label{Section: Stochastic Block Diagram Interpretation}
Consider the stochastic continuous-time setting depicted in Figure~\ref{Fig: Block Diagram Interpretation}(a) satisfying Assumptions~\ref{Ass: M}-\ref{Ass: Uncorrelated}. It is the same as the general setting in Figure~\ref{Fig: MSS Setting}, but it also indicates an It\=o interpretation of the stochastic multiplicative gains. By using the definition of It\=o integrals in Section~\ref{Section: Interpretations}, we construct a discrete-time block diagram, depicted in Figure~\ref{Fig: Block Diagram Interpretation}(b), which explicitly describes the It\=o interpretation of Figure~\ref{Fig: Block Diagram Interpretation}(a). In fact, it is constructed by using a partition $\mathcal P_N[0,t]$ of $N$ subintervals on $[t_0,t_N]:= [0,t]$ as described in Section~\ref{Section: Interpretations}. Therefore, Figure~\ref{Fig: Block Diagram Interpretation}(a) can be interpreted as the limit of Figure~\ref{Fig: Block Diagram Interpretation}(b) as $N \to \infty$. Note that $\mathcal M_N$ denotes a finite dimensional approximation of $\mathcal M$ on the partition $\mathcal P_N[0,t]$, i.e. 
\[ y = \mathcal M_N \tilde u \Longleftrightarrow y_N = \sum_{k=0}^{N-1} M(t_N-t_k) \tilde u_k,\]
where the ``tilde" is used to denote the increments of a signal (refer to Section~\ref{Section: Interpretations}).
\begin{figure}[h!]
	\centering
	\begin{tabular}{ll}
		\includegraphics[scale=0.58]{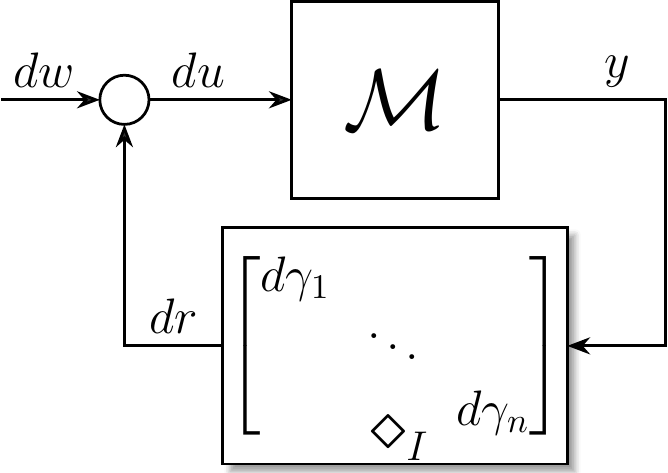} &
		\includegraphics[scale=0.58]{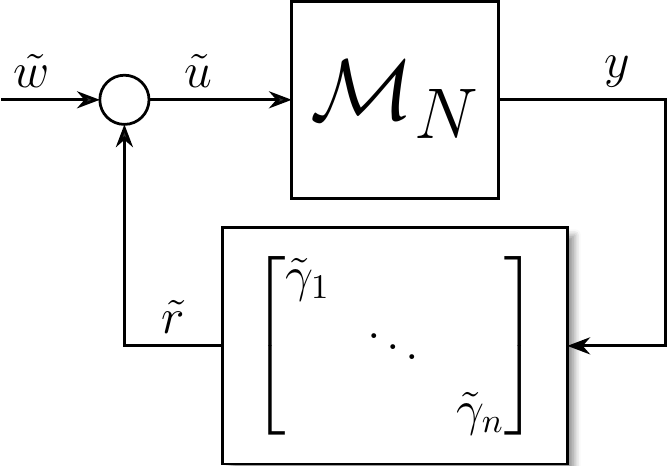} \\
		\footnotesize{(a) Continuous-Time Setting} &
		\footnotesize{(b) Discrete-Time Setting}
	\end{tabular}
	\caption{\footnotesize{A causal LTI system $\mathcal M$ in feedback with stochastic multiplicative gains $\{d\gamma_i\}$ that represent the differential forms of, possibly mutually correlated, Wiener processes. Figure~(a) shows the continuous-time MSS setting when the It\=o interpretation is adopted. Figure~(b) explicitly describes the It\=o interpretation of Figure~(a) by using a partition $\mathcal P_N[0,t]$ of $N$ subintervals as explained in \ref{Section: Interpretations}. In fact, Figure~(a) is interpreted as the limit of Figure~(b) as $N\to\infty$.}}
	\label{Fig: Block Diagram Interpretation}
\end{figure}

The equations describing the block diagrams in Figures~\ref{Fig: Block Diagram Interpretation}(a) and (b) can be respectively written as\\
\begin{subequations}
	\noindent\begin{minipage}{0.2\textwidth}
		\begin{align} \label{Eqn: Ito Interpretation}
		\left\{
		\begin{aligned} 
		y(t) &= \left(\mathcal M du\right)(t)\\
		du(t) &= dw(t) + dr(t)\\
		dr(t) &= d\Gamma(t) \ito y(t)
		\end{aligned}  \right.
		\end{align}
	\end{minipage} \qquad 
	\begin{minipage} {0.2\textwidth}
		\begin{align} \label{Eqn: Equations of MSS Setting Discretized}
		\left\{
		\begin{aligned} 
		y_N &= \left(\mathcal M_N\tilde u\right)_N  \\
		\tilde u_N &= \tilde w_N + \tilde r_N \\
		\tilde r_N &= \tilde \Gamma_N y_N
		\end{aligned} \right.
		\end{align}
	\end{minipage} \\
\end{subequations} 
The rest of this subsection shows that by adopting the It\=o interpretation (\ref{Eqn: Equations of MSS Setting Discretized}), the stochastic signal $r$ will have independent increments. Furthermore, we will derive the expression that describes the propagation of the instantaneous covariance through the feedback block. The analysis is carried out using Figure~\ref{Fig: Block Diagram Interpretation}(b) and then is passed to the limit as $N \to \infty$. 
\subsubsection{Disturbance-to-signals mapping} \quad \\
It is fairly straightforward to show that the disturbance $\tilde w$ is mapped to the various signals in the loop as 
\begin{equation} \label{Eqn: D2S Mapping}
	\begin{bmatrix} \tilde u \\ y \\ \tilde r \end{bmatrix} =
	\begin{bmatrix}
		(I-\tilde\Gamma\mathcal M_N)^{-1} \\ 
		(I-\mathcal M_N\tilde\Gamma)^{-1}\mathcal M_N\\
		(I-\tilde\Gamma\mathcal M_N)^{-1}\tilde\Gamma \mathcal M_N
	\end{bmatrix} \tilde w.	
\end{equation}
\subsubsection{Independence of $\big(d\Gamma(t), y(\tau)\big)$ for $\tau \leq t$} \label{Section: Independence of gamma and v} \quad \\
This can be shown by analyzing the second equation in (\ref{Eqn: D2S Mapping}). Examining the operator $(I - \mathcal M_N \tilde \Gamma)^{-1}$ allows us to write it, over the time horizon of the partition $\mathcal P_N[0,t]$, as
\begingroup\makeatletter\def\f@size{6.5}\check@mathfonts
\def\maketag@@@#1{\hbox{\m@th\large\mnormalfont#1}}%
\begin{align*}
	\begin{bmatrix}  I 				 				& 			& & \\ 
		-M(t_1-t_0) \tilde\Gamma_0 					& 	I		& & \\ 
		& 		\ddots	& \ddots & \\ 
		-M(t_N-t_0) \tilde\Gamma_0 			& \cdots 		& -M(t_N - t_{N-1})\tilde\Gamma_{N-1} & I
	\end{bmatrix}	^{-1} 
	=	\begin{bmatrix}	I & & \\ 
		& \ddots & \\ 
		* &  & I
		\end{bmatrix}	,	
\end{align*} \endgroup
where $*$ denotes the blocks of matrices that are functions of $\tilde \Gamma_k$ for $k=0,1, ..., N-1$. Hence the second equation in (\ref{Eqn: D2S Mapping}) can be written as
\begingroup\makeatletter\def\f@size{6.5}\check@mathfonts
\def\maketag@@@#1{\hbox{\m@th\large\mnormalfont#1}}%
\begin{align*}
	&\begin{bmatrix} 
		y_0  \\ \vdots  \\ y_N  
	\end{bmatrix} ~= 
	\begin{bmatrix}	I & & \\ 
		& \ddots & \\ 
		* &  & I
	\end{bmatrix} 
	\begin{bmatrix}  
		I 				 				& 			& & \\ 
		M(t_1-t_0)  					& 	I		& & \\ 
		& 		\ddots	& \ddots & \\ 
		M(t_N-t_0)  			& \cdots 		& M(t_N-t_{N-1}) & I
	\end{bmatrix}		
	\begin{bmatrix}
		\tilde w_0  \\ \vdots  \\ \tilde w_N    
	\end{bmatrix}.							   
\end{align*} \endgroup
Clearly, $y_N$ does not depend on $\tilde \Gamma_N$ for any positive integer $N$. Furthermore, by carrying out a similar reasoning, it is straightforward to see that $\tilde \Gamma_N$ is independent of the past values of all the signals in the loop (particularly $y$). This analysis shows that $(\tilde \Gamma_N, y_k)$ are independent for $k\leq N$. Finally, taking the limit as $N\to \infty$ completes the argument.
\subsubsection{Temporal independence of the increments of $r$} \label{Section: Independence of increments of r}\quad \\ 
The following calculation shows that $r$ has independent increments. For $k < l$, we have
\begin{align*}
	\expec{\tilde r_k \tilde r^*_l} &= \expec{\tilde \Gamma_k y_k y_l^* \tilde \Gamma^*_l}
	= \expec{\tilde \Gamma_k y_k y_l^*} \expec{\tilde \Gamma^*_l}
	= 0,
\end{align*}
where the third equality holds because $\tilde \Gamma$ has a zero-mean, and the second equality follows because $\Gamma$ has independent increments (Wiener process) and also $\tilde \Gamma$ is independent of present and past values of $y$ (Section~\ref{Section: Independence of gamma and v}). 

The combination between the causality of $\mathcal M$ and the It\=o interpretation introduces a sort of ``strict causality" in continuous-time systems. Thus the multiplicative, temporally independent gains $\{d\gamma_i(t)\}$ has a ``whitening" effect. In fact, although $y$ has nonzero temporal correlations, the signal $r$ is guaranteed to have independent increments $dr$, i.e. $\expec{dr(t) dr^*(\tau)} = 0, ~\forall t \neq \tau$.

Finally, the instantaneous covariance of $dr$ is calculated as
\begin{align*} 
	\expec{dr(t) dr(t)^*} &= \expec{d\Gamma(t) y(t) y^*(t) d\Gamma^*(t)} \\
	&= \mathbb E \bigg[d\Gamma(t) \expec{y(t)y^*(t)} d\Gamma^*(t) \bigg]\\
	&=  \cov{\Gamma} \circ\cov Y(t) dt =: \cov R(t) dt,
\end{align*}
where the second equality is a consequence of Lemma~\ref{Lemma: Matrix Indp} since $d\Gamma(t)$ and $y(t)$ are independent (Section~\ref{Section: Independence of gamma and v}). The third equality is an immediate consequence of the fact that $d\Gamma(t) = \mathcal D\big(d\gvec(t)\big)$.
Finally, we have 
\begin{equation} \label{Eqn: Hadamard Covariance}
	\cov R(t) = \cov \Gamma \circ \cov Y(t).
\end{equation}

\subsection{Covariance Feedback System} \label{Section: Covariance Feedback System}
The goal of this section is to construct a deterministic feedback system that describes the evolution of the instantaneous covariance matrices of the various signals in Figure~\ref{Fig: Block Diagram Interpretation} and finally derive the expression of the LGO given in (\ref{Eqn: Ito LGO}). 

In the previous section, we showed that $r$ has temporally independent increments. As a result, it is straightforward to see that $u$ also has temporally independent increments, because for $k < l$ we have
\begin{align*}
	\expec{\tilde u_k \tilde u_l^*} &= \expec{(\tilde w_k + \tilde r_k)(\tilde w_l + \tilde r_l)^*} \\
	&= \expec{\tilde w_k \tilde w_l^*} + \expec{\tilde r_k \tilde r_l^*} + \expec{\tilde r_k \tilde w_l} + \expec{\tilde w_k \tilde r_l^*}\\
	&= 0 + 0 + 0 + \expec{\tilde w_k y_l^* \tilde \Gamma_l^*} \\
	&= \expec{\tilde w_k y_l^*} \expec{\tilde \Gamma_l^*} = 0,
\end{align*}
where the third equality follows from the fact that $w$ (Wiener process) and $r$ (Section~\ref{Section: Independence of increments of r}) both have independent increments and the fact that $w$ is independent of past values of all the signals in the loop. The fourth equality follows from Section~\ref{Section: Independence of gamma and v} and the assumption that $w$ and $\Gamma$ are independent. Finally, passing to the limit as $N\to \infty$ yields that $du$ is temporally independent. 

As for the instantaneous covariance of $\tilde u$, we have
\begin{align*}
	\expec{\tilde u_k \tilde u_k^*} &= \expec{\tilde w_k \tilde w_k^*} + \expec{\tilde r_k \tilde r_k^*} + \expec{\tilde r_k \tilde w_k^*} + \expec{\tilde w_k \tilde r_k^*}\\
	&= \cov W_k \Delta_k  + \cov R_k \Delta_k + \expec{\tilde \Gamma_k y_k \tilde w_k^*} + \expec{\tilde w_k y_k^* \tilde \Gamma_k^*} \\
	& = (\cov W_k + \cov R_k) \Delta_k + 0 + 0 	=: \cov U_k \Delta_k.
\end{align*}
Therefore, the addition junction in Figure~\ref{Fig: Block Diagram Interpretation} remains as an addition operation on the associated covariance matrices, i.e.
\begin{align} \label{Eqn: Covaraince Additions}
	\cov U(t) = \cov W(t) + \cov R(t).
\end{align}
Furthermore, the propagation of the covariance through the forward block of Figure~\ref{Fig: Block Diagram Interpretation} is given by (\ref{Eqn: IO Cov}) which requires the input $du$ to be temporally independent for its validity. Finally, the propagation of the covariance through the feedback block is given by (\ref{Eqn: Hadamard Covariance}). Therefore, (\ref{Eqn: IO Cov}), (\ref{Eqn: Hadamard Covariance}) and (\ref{Eqn: Covaraince Additions}) can be used to construct the deterministic feedback block diagram depicted in Figure~\ref{Fig: Covariance Diagram}, where each signal is matrix-valued.
\begin{figure}[h!]
	\centering
	\includegraphics[scale=0.5]{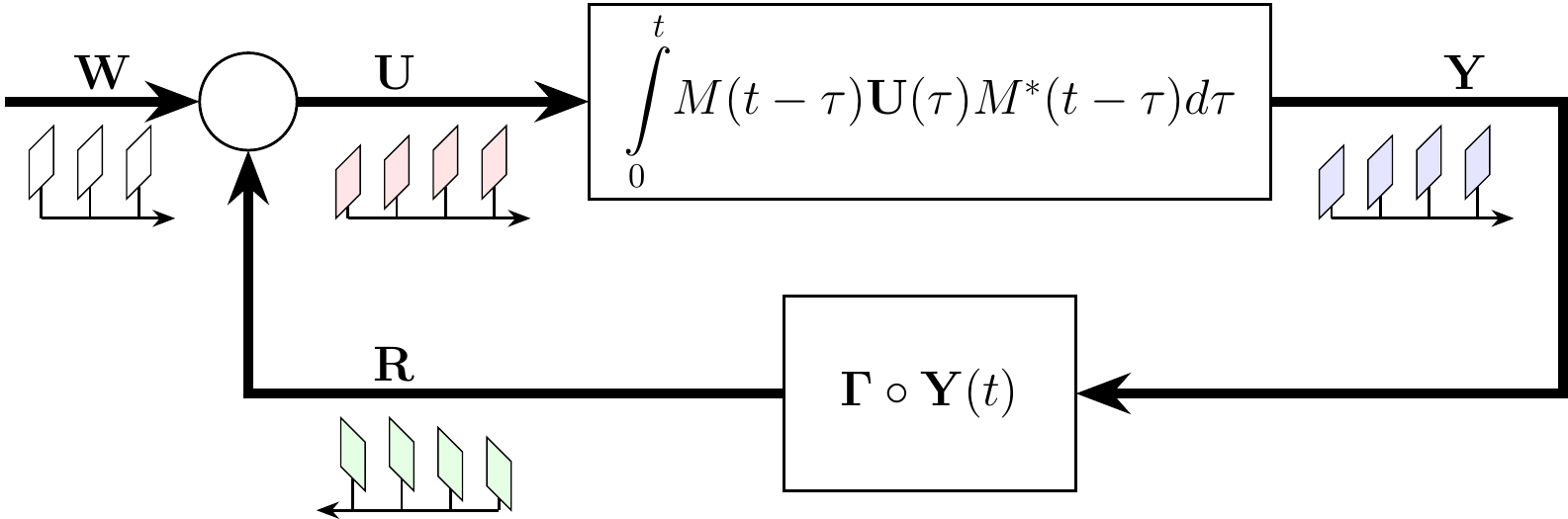} 
	\caption{\footnotesize{A deterministic block diagram describing the evolution of the covariance matrices of the various signals in the feedback loop of Figure~\ref{Fig: Block Diagram Interpretation}(a). The forward block represents a convolution integral of matrices and the feedback block represents a Hadamard (element-by-element) product. Note that all the covariance matrices in the loop are positive semi-definite and non-decreasing in time when $\cov W$ is non-decreasing, i.e. for $t_2 \geq t_1$, $\cov W(t_2) - \cov W(t_1) \geq 0$ (refer to \cite{bamieh2018structured}).}}
	\label{Fig: Covariance Diagram}
\end{figure}
The advantage of the covariance feedback system in Figure~\ref{Fig: Covariance Diagram} is that it describes a deterministic dynamical system unlike its corresponding stochastic feedback system in Figure~\ref{Fig: Block Diagram Interpretation}. Before we construct the loop gain operator, we give a remark.
\begin{remark} \label{Remark: Monotone Covariances}
	All the covariance signals in Figure~\ref{Fig: Covariance Diagram} are monotone. Particularly, if $t_1 \leq t_2$ then $\cov U(t_1) \leq \cov U(t_2)$, where the matrix ordering is taken in the usual positive semidefinite sense. Refer to \cite[Section II-E]{bamieh2018structured}.
\end{remark}
\subsection{Loop Gain Operator} \label{Section: LGO}
We are now equipped with all the necessary tools to define the continuous-time counterpart of the LGO introduced in \cite{bamieh2018structured}. Over a finite time horizon $[0,t]$, the instantaneous covariance $\cov R(t)$ can be expressed in terms of $\{\cov U(\tau), 0\leq \tau \leq t\}$ using (\ref{Eqn: IO Cov}) and (\ref{Eqn: Hadamard Covariance}) as
\begin{align} \label{Eqn: RU Relationship}
	\cov R(t) &= \cov \Gamma \circ \cov Y(t) \nonumber\\
	&= \cov \Gamma \circ \left(\int_0^t M(t-s) \cov U(s) M(t-s)ds \right) \nonumber \\
	\cov R(t) &= \cov \Gamma \circ \left(\int_0^t M(\tau) \cov U(t-\tau) M^*(\tau) d\tau\right).
\end{align}
The previous calculation motivates the definition of a finite dimensional linear operator over the infinite time horizon, i.e. as $t\to \infty$
\begin{align} \label{Eqn: ITH LGO}
\sscov R = \mathbb L\left(\sscov U\right) &:= \cov \Gamma \circ \left(\int_0^\infty M(\tau) \sscov U M^*(\tau) d\tau\right)
\end{align}
where $\sscov U$ and $\sscov R$ are the steady-state limits (if they exist) of the covariances. This linear operator acts on a matrix to produce another matrix, and it propagates the steady state covariance $\sscov U$ ``once around the loop" to produce the steady state covariance $\sscov R$ (and thus the name \textit{loop gain operator}, refer to Figure~\ref{Fig: Covariance Diagram}).  Before moving to the next section, we define here a truncated version of the LGO as
\begin{equation} \label{Eqn: Truncated LGO}
	\mathbb L_T\left(X\right) := \cov \Gamma \circ \left(\int_0^T M(\tau) X M^*(\tau) d\tau\right),
\end{equation}
which will be useful when proving Theorem~\ref{Thm: MSS Conditions}. Before stating the proof, we summarize some useful properties of the LGO in three remarks.
\begin{remark} \label{Remark: Monotone Operator}
 	The operator $\mathbb L_T$ defined in (\ref{Eqn: Truncated LGO}) is a monotone operator, i.e. if $0 \leq X \leq Y$, then $0 \leq \mathbb L_T(X) \leq \mathbb L_T(Y)$. The same property holds for $\mathbb L$ defined in $(\ref{Eqn: ITH LGO})$ since $\mathbb L = \lim_{T\to\infty} \mathbb L_T$. Refer to \cite[Section II-E]{bamieh2018structured} for details, noting that the same arguments also hold for integrals as well as summations.
 \end{remark}
 \begin{remark} \label{Remark: Monotone Operator in Time}
 	The operator $\mathbb L_T$ is also monotone in time, i.e. if $T_1 \leq T_2$, then $0 \leq \mathbb L_{T_1}(X) \leq \mathbb L_{T_2} (X)$ for any $X \geq 0$.
 	This is easy to validate by checking that $\mathbb L_{T_2}(X) - \mathbb L_{T_1}(X)$ is positive semidefinite. Consequently, for any $T>0$ and $X \geq0$, we have $0 \leq \mathbb L_T(X) \leq \mathbb L(X)$.
 \end{remark}
 \begin{remark} \label{Remark: Perron Frobenius}
 	The spectral radius of $\mathbb L$ is its largest eigenvalue which is guaranteed to be a real number. Furthermore, the ``eigen-matrix" associated with the largest eigenvalue is guaranteed to be positive semidefinite. That is, if $\rho(\mathbb L)$ denotes the spectral radius of $\mathbb L$, then $\exists \ecov U \geq 0$ s.t. $\mathbb L(\ecov U) = \rho(\mathbb L) \ecov U$.
 	Note that $\ecov U$ is the matrix counterpart of the Perron-Frobenius vector for matrices with nonnegative entries. This is the covariance mode that has the fastest growth rate if MSS is violated, and therefore we refer to $\ecov U$ as the worst-case covariance. (Refer to \cite[Thm 2.3]{bamieh2018structured} for more details.)
 \end{remark}
\subsection{MSS Conditions} \label{Section: MSS Conditions Proof}
Equipped with the LGO, we can now present the proof of Theorem~\ref{Thm: MSS Conditions}. The proof is very similar to the discrete-time counterpart in \cite{bamieh2018structured}, and thus some of the details are omitted. 
\begin{proof} 
\subsubsection{if} 
Using (\ref{Eqn: Covaraince Additions}) and (\ref{Eqn: RU Relationship}), $\cov U(t)$ can be written as
\begin{align*}
	\cov U(t) &= \cov \Gamma \circ \left(\int_0^t M(\tau) \cov U(t-\tau) M^*(\tau) d\tau\right) + \cov W(t)\\
	&\leq \cov \Gamma \circ \left(\int_0^t M(\tau) \cov U(t) M^*(\tau) d\tau\right) + \cov W(t)\\
	& \leq \mathbb L \Big(\cov U(t)\Big) + \cov W(t),
\end{align*}
where the first inequality follows from Schur's theorem \cite[Thm 2.1]{horn1992block} and the fact that $\cov U(t-\tau) \leq \cov U(t)$ for all $\tau \in[0,t]$ (Remark~\ref{Remark: Monotone Covariances}). The second inequality follows from Remark~\ref{Remark: Monotone Operator in Time}. To obtain an upper bound on $\cov U(t)$, we let $\mathbb I$ denote the identity operator and rearrange to obtain
\begin{align*} 
	(\mathbb I - \mathbb L) \cov U(t) &\leq \cov W(t) \leq \sscov W\\
	\cov U(t) &\leq (\mathbb I - \mathbb L)^{-1} \sscov W, 
\end{align*}
where the second equality is obtained by replacing $\cov W(t)$ with its steady state value $\sscov W$ since it is assumed to be monotone (Assumption~\ref{Ass: w}). The third inequality is obtained by applying \cite[Thm 2.3]{bamieh2018structured} which guarantees that the operator $(\mathbb I - \mathbb L)^{-1}$ exists and is monotone whenever $\mathbb L$ is monotone and $\rho(\mathbb L) < 1$. Finally the stability of $\mathcal M$ (finite $H^2-\text{norm}$) guarantees that all other covariance signals in the loop of Figure~\ref{Fig: Covariance Diagram} are also uniformly bounded thus guaranteeing MSS.
\subsubsection{only if} First it is straightforward to show that MSS is lost if the $H^2$-norm of $M$ is infinite (regardless of the value of $\rho(\mathbb L))$. Using Figure~\ref{Fig: Covariance Diagram}, we can write the covariance $\cov Y(t)$ as
\begin{align*}
	\cov Y(t) &= \int_0^t M(t-\tau) \cov U(\tau) M^*(t-\tau)d\tau \\
	&= \int_0^t M(t-\tau) \Big(\cov W(\tau) + \cov \Gamma \circ \cov Y(\tau) \Big) M^*(t-\tau)  d\tau \\
	& \geq \int_0^t M(t-\tau) \cov W(\tau) M^*(t-\tau) d\tau,
\end{align*}
where the inequality follows from the fact that $\cov \Gamma \circ \cov Y(\tau)$ is positive semidefinite. Thus, clearly $\cov Y(t)$ grows unboundedly when $M$ has an infinite $H^2$-norm (take $\cov W(t) = I$ for example).

Next, assume that $M$ has a finite $H^2$-norm. We will show that if $\rho(\mathbb L) \geq 1$, then $\cov U(t)$ grows unboundedly in time. We do so by examining $\cov U(t)$ at the time samples $t_k:= kT$, where $k$ is a positive integer and $T > 0$. Using Figure~\ref{Fig: Covariance Diagram}, we obtain
\begin{align} \label{Eqn: Recursion}
	\cov U(t_k) &= \cov \Gamma \circ \int_0^{t_k} M(t_k - \tau) \cov U(\tau) M^*(t_k - \tau) d\tau + \cov W(t_k) \nonumber\\
	& \geq  \cov \Gamma \circ \int_{t_{k-1}}^{t_k} M(t_k - \tau) \cov U(\tau) M^*(t_k - \tau) d\tau + \cov W(t_k) \nonumber\\
	& \geq \cov \Gamma \circ \int_{t_{k-1}}^{t_k} M(t_k - \tau) \cov U(t_{k-1}) M^*(t_k - \tau) d\tau + \cov W(t_k) \nonumber \\
	& \geq \cov \Gamma \circ \int_0^T M(s) \cov U(t_{k-1}) M^*(s) ds + \cov W(t_k) \nonumber\\
	& = \mathbb L_T \Big(\cov U(t_{k-1}) \Big) + \cov W(t_k) \nonumber\\
	\cov U(t_k)  &\geq \mathbb L_T^k\Big( \cov U(0) \Big) + \sum_{r=0}^{k-1} \mathbb L_T^r \Big( \cov W(t_{k-r})\Big), 
\end{align}
where the first inequality follows from the fact that the integrand is positive semidefinite, the second inequality follows because $\cov U(\tau) \geq \cov U(t_{k-1})$ for $\tau \in [t_{k-1}, t_k]$, and the third inequality is a consequence of applying the change of variable $s:=t_k - \tau$. The last inequality is a consequence of a simple induction argument that exploits the monotonicity of $\mathbb L_T$ (Remark~\ref{Remark: Monotone Operator}). Establishing the inequality (\ref{Eqn: Recursion}) allows us to use the same arguments in \cite{bamieh2018structured} (repeated here for completeness) to show that $\cov U(t_k)$ grows unboundedly. 

Set the exogenous covariance $\cov W(t_k) = \ecov U$, where $\ecov U$ is the worst-case covariance described in Remark~\ref{Remark: Perron Frobenius}. Note that the initial covariance is $\cov U_0 = \ecov U$. 
Substituting in (\ref{Eqn: Recursion}) yields
\begin{align} \label{Eqn: Recursion IC}
	\cov U(t_k) \geq \sum_{r=0}^k \mathbb L_T^r \left(\ecov U \right).
\end{align}
Since $\lim_{T\rightarrow\infty} \mathbb L_T(\ecov U) =\mathbb L(\ecov U) = \rho(\mathbb L) \ecov U$,  then for any ${\epsilon>0}$, 
$\exists ~T>0$ such that $||\rho(\mathbb L) \ecov U - \mathbb L_T(\ecov U)|| \leq \epsilon ||\ecov U||$. This inequality coupled with the fact that $0 \leq \mathbb L_T(\ecov U) \leq \rho(\mathbb L) \ecov U$ allows us to invoke \cite[Lemma A.3]{bamieh2018structured} to obtain
\begin{equation}\label{Eqn: Bound on LT}
	\mathbb L_T(\ecov U)     ~\geq~     \left( \rho(\mathbb L) - \epsilon c\right)   ~ \ecov U ~=:~    \alpha~ \ecov U, 
\end{equation}
where $c$ is a positive constant that only depends on $\ecov U$. 
Then, by (\ref{Eqn: Recursion}), the one-step lower bound~(\ref{Eqn: Bound on LT}) becomes
\begin{equation}\label{Eqn: U Growth}
	\cov U(t_k)~\geq~ \left( \sum_{r=0}^k \alpha^r \right) ~ \ecov U
	~=~ \frac{\alpha^{k+1}-1}{\alpha-1} ~ \ecov U.
\end{equation}
First consider the case when $\rho(\mathbb L)>1$, then $\epsilon$ can be chosen small enough so that $\alpha>1$ and therefore $\{\ecov U(t_k)\}$ is a geometrically growing 
sequence. As for the case where $\rho(\mathbb L) = 1$, we have $\alpha = 1 - \epsilon$. Then for $0 < \epsilon < 1$, we have
$$ \sscov U = \lim_{k\to \infty} \cov U(t_k) \geq \frac{1}{\epsilon} \ecov U.$$
This proves that $\cov U(t)$ can grow arbitrarily large (although not necessarily geometrically) since $\epsilon$ can be chosen to be arbitrarily small.
\end{proof}

\section{Conclusion} \label{Section: Conclusion}
This paper examines the conditions of MSS for LTI systems in feedback with multiplicative stochastic gains. The analysis is carried out from a purely-input output approach as compared to (the more common) state space approach in the literature. The advantage of this approach is encompassing a wider range of models. It is shown that in the continuous-time setting, technical subtleties arise that require to exploit several tools from stochastic calculus. Different stochastic interpretations are considered for which different stochastic block diagram representations are constructed. Finally, it is shown that MSS analysis for state space realizations can be transparently carried out as a special case of our approach.

\section*{Acknowledgments}
The authors would like to thank Professor Jean-Pierre Fouque for the valuable discussions on stochastic calculus.

\bibliographystyle{ieeetr}
\bibliography{CT_StochUnc_References}

\appendix

\subsection{Interpretations of Stochastic Convolution} \label{Section: Stochastic Convolution}
Consider the stochastic convolution in (\ref{Eqn: Stochastic Convolution}) satisfying Assumption~\ref{Ass: M}. Exploiting the partition $\mathcal P_N[0,t]$ described in Section~\ref{Section: Partition} and the notation developed in Section~\ref{Section: Increments} yield
$$ y(t) = \lim_{N\to \infty} \sum_{k=0}^{N-1} M(t-\bar t_k) \tilde u_k, $$
where $\bar t_k \in [t_k, t_{k+1}]$. The choice of $\bar t_k$ prescribes a particular stochastic interpretation of the integral, for example $\bar t_k = t_k$ corresponds to an It\=o interpretation. The following calculation shows that the covariance of $y$ does not depend on the choice of $\bar t_k$ when $M \in \mathcal C$ defined in Appendix~\ref{Section: Variations of M}.
\begin{align*}
	\cov Y(t) &:= \expec{y(t)y^*(t)} \\
	&= \lim_{N\to \infty} \sum_{k,l=0}^{N-1} M(t-\bar t_k)\expec{\tilde u_k \tilde u_l^*} M^*(t-\bar t_l) \\
	&= \lim_{N\to \infty} \sum_{k=0}^{N-1} M(t-\bar t_k) \expec{\tilde u_k \tilde u_k^*} M^*(t-\bar t_k) \\
	&= \lim_{N\to \infty} \sum_{k=0}^{N-1} M(t-\bar t_k) \cov U(t_k) \Delta_k M^*(t-\bar t_k) \\
	&= \int_0^t M(t-\tau) \cov U(\tau) M^*(t-\tau) d\tau,
\end{align*}
where the third equality follows from the temporal independence of $u$ and the fourth equality follows from the definition of the covariance of $du$. The last equality is a consequence of Riemann integrability which guarantees convergence to a unique value when $M \in \mathcal C$. As a result, there is no need to prescribe a stochastic interpretation of (\ref{Eqn: Stochastic Convolution}) since different stochastic interpretations  play the same role in the mean-square sense.  
\subsection{Calculation of $D_N(t)$ in (\ref{Eqn: Def DN})} \label{Section: Algebraic Manipulations}
This appendix shows the required algebraic manipulations to arrive at the expression of $D_N(t)$ in (\ref{Eqn: DN}). Start by adding and subtracting $M(t-t_k) \tilde \Gamma_k y_k$ in the partial sum of $S_N(t)$ in (\ref{Eqn: Integrals Definitions}) to obtain 
\begingroup\makeatletter\def\f@size{9}\check@mathfonts
\def\maketag@@@#1{\hbox{\m@th\large\mnormalfont#1}}
\begin{align*} 
	S_N(t) = I_N(t) + \frac{1}{2} \sum_{k=0}^{N-1} \bigg(M(t-t_{k+1}) \tilde \Gamma_k y_{k+1} - M(t-t_k) \tilde \Gamma_k y_k \bigg) ,
\end{align*}\endgroup
where $I_N(t)$ is defined in (\ref{Eqn: Integrals Definitions}). Adding and subtracting ${M(t-t_{k+1}) \tilde \Gamma_k y_k}$ in the sum of the second term yields
\begin{align} \label{Eqn: SISO Discretized Strat Integral}
	S_N(t) &= I_N(t) + \frac{1}{2} \big(Q_N(t) + J_N(t)\big),
\end{align}
where $J_N(t)$ is given in (\ref{Eqn: Sums}) and
\begin{align} 
	Q_N(t) &:= \sum_{k=0}^{N-1} M(t-t_{k+1}) \tilde \Gamma_k \tilde y_k \label{Eqn: QN First Form}
\end{align}
Observe that $Q_N(t)$ (\ref{Eqn: QN First Form}) is a cross quadratic-variation-like term whose limit is not obvious, so we examine the increments $\tilde y_k$ using (\ref{Eqn: Single SIE}) with $\diamond = \strat$. We have
\begin{align} \label{Eqn: y Increments}
	\tilde y_k &= E_{k+1}(t_{k+1}) - E_{k}(t_k) + S_{k+1}(t_{k+1}) - S_k(t_k) \nonumber \\
	\tilde y_k &=:\tilde E_k + \tilde I_k + \frac{1}{2}\left( \tilde Q_k + \tilde J_k\right).
\end{align}
where $E_N(t): = \sum\limits_{k=0}^{N-1} M(t-t_k) \tilde w_k$. 
Start by calculating $\tilde E_k$
\begin{align*}
	\tilde E_k &= \sum_{l=0}^{k} M(t_{k+1} - t_l) \tilde w_l - \sum_{l=0}^{k-1} M(t_{k} - t_l) \tilde w_l \\
	&= M(\Delta_k) \tilde w_k  + \sum_{l=0}^{k-1} \bigg(M(t_{k+1}-t_l) - M(t_k-t_l)\bigg) \tilde w_l.
\end{align*}
Carrying out similar calculations for $\tilde I_k, \tilde Q_k$ and $\tilde J_k$ yields
\begin{align*}
	\tilde I_k &= M(\Delta_k) \tilde \Gamma_k y_k  + \sum_{l=0}^{k-1} \bigg(M(t_{k+1}-t_l) - M(t_k-t_l)\bigg) \tilde \Gamma_l y_l \\
	\tilde Q_k &= M_0 \tilde \Gamma_k \tilde y_k + \sum_{l=0}^{k-1} \bigg(M(t_{k+1} - t_{l+1}) - M(t_k - t_{l+1}) \bigg) \tilde \Gamma_l \tilde y_l\\
	\tilde J_k &= \bigg(M_0 - M(\Delta_k) \bigg)  \tilde \Gamma_k y_k +
	\sum_{l=0}^{k-1} \bigg( M(t_{k+1}- t_{l+1})  \\
	& \qquad  - M(t_k-t_{l+1}) + M(t_k-t_l) - M(t_{k+1}-t_l) \bigg)  \tilde \Gamma_l y_l,
\end{align*}
where $M_0$ denotes $M(0)$ for notational brevity. Substituting for the expression of $\tilde y_k$ (\ref{Eqn: y Increments}) in $Q_N(t)$ (\ref{Eqn: QN First Form}) and collecting terms yield
\begin{align*} 
	Q_N(t) &= \frac{1}{2} \Big(\theta_N(t) + \eta_N(t)  + T_N^\alpha (t) + T_N^\beta (t)\Big) \nonumber \\
	& \quad + \chi_N(t) + T^\zeta_N(t) + \sum_{k=0}^{N-1} M(t-t_{k+1})\tilde \Gamma_k M_0 \tilde \Gamma_k y_k,
\end{align*}
where $\theta_N(t), \eta_N(t), \chi_N(t), T_N^\alpha (t), T_N^\beta(t)$ and $T_N^\zeta (t)$ are all defined in (\ref{Eqn: Sums}). Adding and subtracting $ M(t-t_k) \tilde \Gamma_k M_0 \tilde \Gamma_k y_k$ in the partial sum of the last term yields 
\begin{align} \label{Eqn: QN Second Form}
	Q_N(t) &= \frac{1}{2} \Big(\theta_N(t) + \eta_N(t) + T_N^\alpha (t) + T_N^\beta (t) \Big) + \nu_N(t) \nonumber \\
	& + \chi_N(t) + T_N^\zeta(t) + \sum_{k=0}^{N-1} M(t-t_k)\tilde \Gamma_k M_0 \tilde \Gamma_k y_k,
\end{align}
where $\nu_N(t)$ is defined in (\ref{Eqn: Sums}). Finally, $D_N(t)$ is calculated as
\begin{align} \label{Eqn: SISO DN Form1}
	D_N(t) &:= S_N(t) - \left(I_N(t) + \frac{1}{2} R_N(t) \right) \nonumber \\
	&= \frac{1}{2} \bigg(Q_N(t) - R_N(t) + J_N(t)\bigg).
\end{align}
Substituting for $Q_N(t)$ from (\ref{Eqn: QN Second Form}), $R_N(t)$ from (\ref{Eqn: Integrals Definitions}), and $J_N(t)$ from (\ref{Eqn: Sums}), yields the expression of $D_N(t)$ given in (\ref{Eqn: DN}) after exploiting the following equation
$$ \tilde \Gamma_k M_0 \tilde \Gamma_k - (M_0 \circ \cov \Gamma) \Delta_k = \Big(\tilde{\boldsymbol \gamma}_k \tilde{\boldsymbol \gamma}_k^* - \cov \Gamma \Delta_k \Big) \circ M_0,$$
where $\tilde{\boldsymbol \gamma}_k = \mathcal D(\Gamma_k)$ is the vector formed of the diagonal entries of $\Gamma_k$.

\subsection{Second Moments of Cross Terms} \label{Section: Cross Terms}
Let $x$ and $y$ be two vector-valued random variables. The subsequent calculation shows that to check if $\expec{\vnorm{x+y}^2}$ is zero, it suffices to check that $\expec{\vnorm{x}^2} = \expec{\vnorm{y}^2} = 0$. 
\begingroup\makeatletter\def\f@size{9}\check@mathfonts
\def\maketag@@@#1{\hbox{\m@th\large\mnormalfont#1}}
\begin{align*}
	\expec{ \vnorm{x + y}^2 } & \leq \expec{\left( \vnorm{x} + \vnorm{y} \right)^2} \\
	&= \expec{ \vnorm{x}^2 + \vnorm{y}^2 + 2 \vnorm{x}\vnorm{y} } \\
	& \leq \expec{ \vnorm{x}^2 }  + \expec{\vnorm{y}^2} + 2 \sqrt{ \expec{\vnorm{x}^2} \expec{\vnorm{y}^2} },
\end{align*}\endgroup
where the first inequality is a consequence of applying the triangle inequality, and the last one follows from Cauchy-Schwarz inequality with respect to expectations. Observe that if $\expec{ \vnorm{x}^2 }$ or $\expec{ \vnorm{y}^2 }$ is zero, then the cross term is zero. Therefore, to prove that the variance of the sum of random variables is equal to zero, there is no need to calculate the expectation of cross terms.

\subsection{Useful Equalities \& Inequalities} \label{Section: Useful Lemmas on Inequalities}
This appendix provides a sequence of lemmas that give some useful equalities and inequalities (upper bounds) that are used in the proofs throughout the paper. 
\begin{lemma} \label{Lemma: Matrix Indp}
	Let $X$ and $v$ be a matrix-valued and vector-valued random variables, respectively. If $X$ and $v$ are independent and $D_v:= \mathcal D(v)$, then 
	$$ \expec{D_vXD_v} = \expec{vv^*} \circ \expec{X}.$$
\end{lemma}
\begin{proof}
	Let $X_{ij}$ denote the $ij^{\text{th}}$ entry of the matrix $X$. Then
	\begin{align*}
		\expec{D_v X D_v}_{ij} &= \expec{v_i X_{ij} v_j} = \expec{v_i v_j} \expec{X_{ij}} \\
		&= \expec{v v^*}_{ij} \expec{X}_{ij},
	\end{align*}
	where the first equality holds because $D_v:= \mathcal D(v)$ is diagonal, and the second equality hold because $X$ and $v$ are independent. The proof is complete since the Hadamard product ``$\circ$" is the element-by-element multiplication.
\end{proof}
\begin{lemma} \label{Lemma: Moments Normal Vector}
	Let $x = \begin{bmatrix} x_1 & x_2 & \cdots & x_n \end{bmatrix}^*$ be a zero-mean random vector that follows a multivariate normal distribution with a covariance matrix $\cov \Sigma := \expec{xx^*}$. Then
	$$ \expec{\vnorm{x}^2} = \text{tr}(\cov \Sigma) \quad \text{and} \quad \expec{\vnorm{x}^{2p}}  \leq c(p,n)\mnorm{\cov \Sigma}^p,$$
	where $p$ is any positive integer and $c$ is a constant that depends on $p$ and $n$. For example, one can check that $c(1,n) = n$ and $c(2,n) = n^2 + 2n$.
\end{lemma}
\begin{proof}
	For the second moment, we have
	\begin{align*}
		\expec{\vnorm{x}^2} &= \sum_{i=1}^n \expec{x_i^2} = \sum_{i=1}^n \cov \Sigma_{ii} = \text{tr}(\cov \Sigma). 
	\end{align*}
	To calculate the fourth moment, let $\cov \Sigma^{1/2}$ denote the Cholesky factorization of $\cov \Sigma$ so that $x = \cov \Sigma^{1/2} \xi$ where $\xi$ follows the standard multivariate normal distribution. Then
	\begin{align*}
		\expec{\vnorm{x}^{2p}} &= \expec{\vnorm{\cov \Sigma^{\frac{1}{2}} \xi}^{2p}} 
		\leq \mnorm{\cov \Sigma}^p \expec{\vnorm{\xi}^{2p}}  \\
		&= \mnorm{\cov \Sigma}^p \expec{\left( \sum_{i=1}^n \xi_i^2 \right)^p} \\
		&= \mnorm{\cov \Sigma}^p \expec{ \sum_{k_1 +k_2+ \cdots + k_n = p} p! \prod_{i=1}^n \frac{\xi_i^{2 k_i}}{k_i!} } \\
		&= \mnorm{\cov \Sigma}^p \sum_{k_1 +k_2+ \cdots + k_n=p} p! \prod_{i=1}^n \frac{\expec{\xi_i^{2 k_i}}}{k_i!}  \\
		&= \mnorm{\cov \Sigma}^p p!\sum_{k_1 +k_2+ \cdots + k_n=p}  \prod_{i=1}^n \frac{\left(2k_i-1\right)!!}{k_i!} \\ &=: c(p,n) \mnorm{\cov \Sigma}^p  ,
	\end{align*}
	where ``$!!$" is the double factorial operation. The inequality follows from the sub-multiplicative property of the norms, the third equality is a direct application of the multinomial theorem, and the fourth equality holds because $\{\xi_i\}$ are mutually independent. Finally, the fifth equality follows because the $m^{\text{th}}$ moment of a standard normal random variable is $(m-1)!!$ when $m$ is even. 
\end{proof}

Throughout Lemmas \ref{Lemma: Ineq Triangle Sub-multiplicative}-\ref{Lemma: Ineq Ind zero}, let $\{X_k\}$ and $\{y_k\}$ be two sequences of square random matrices and random vectors, respectively, with bounded second moments. Furthermore, let $\{F_k\}$ be a sequence of deterministic matrices.
\begin{lemma} \label{Lemma: Ineq Triangle Sub-multiplicative}
	Exploiting the triangle inequality and the sub-multiplicative property of the norm yields
	$$ \expec{\mnorm{\sum_{k=0}^{N-1} F_k X_k y_k}^2} \leq \expec{ \left( \sum_{k=0}^{N-1} \mnorm{F_k}\mnorm{X_k} \vnorm{y_k}  \right)^2 }. $$
\end{lemma}
\begin{lemma} \label{Lemma: Ineq Dep}
	Suppose that $(X_k,y_k)$ are in general dependent, but $\{X_k\}$ has independent increments, i.e. $(X_k,X_l)$ are independent for $k\neq l$. Then \\
	\begin{align*}
		\expec{\vnorm{\sum_{k=0}^{N-1} F_k X_k y_k}^2} 
		&\leq \sum_{k=0}^{N-1} \mnorm{F_k}^2 \sqrtp{ \expec{\mnorm{X_k}^4}} \\
		& \qquad  \times \sqrtp{\expec{ \left( \sum_{k=0}^{N-1}\vnorm{y_k}^2 \right)^2 }} .
	\end{align*}
\end{lemma}
\begin{proof}
	\begin{align*}
		& \expec{\mnorm{\sum_{k=0}^{N-1} F_k X_k y_k}^2} \leq \expec{ \left( \sum_{k=0}^{N-1} \mnorm{F_k}\mnorm{X_k} \vnorm{y_k}  \right)^2 } \\
		& \quad\leq \expec{ \sum_{k=0}^{N-1} \mnorm{F_k}^2 \mnorm{X_k}^2 \sum_{k=0}^{N-1} \vnorm{y_k}^2 } \\
		& \quad \leq \sqrtp{ \expec{ \left( \sum_{k=0}^{N-1} \mnorm{F_k}^2\mnorm{X_k}^2 \right)^2 } \expec{ \left( \sum_{k=0}^{N-1}\vnorm{y_k}^2 \right)^2 }}
	\end{align*}
	where the first inequality follows from Lemma~\ref{Lemma: Ineq Triangle Sub-multiplicative}, the second follows by applying the Cauchy-Schwarz inequality, and the last one follows by applying  again the Cauchy-Schwarz inequality but with respect to the expectation. To complete the proof, we find a bound on the first term of the last inequality. We have
	\begin{align*}
		&\expec{ \left( \sum_{k=0}^{N-1} \mnorm{F_k}^2\mnorm{X_k}^2 \right)^2} \\
		& \qquad  = \sum_{k,l=0}^{N-1} \mnorm{F_k}^2 \mnorm{F_l}^2 \expec{\mnorm{X_k}^2 \mnorm{X_l}^2} \\
		& \qquad \leq \sum_{k,l=0}^{N-1} \mnorm{F_k}^2 \mnorm{F_l}^2 \sqrtp{ \expec{\mnorm{X_k}^4} \expec{\mnorm{X_l}^4}} \\
		& \qquad \leq \left(\sum_{k=0}^{N-1} \mnorm{F_k}^2 \sqrtp{ \expec{ \mnorm{X_k}^4 } } \right)^2,
	\end{align*}
	where the first inequality is obtained by using the Cauchy-Schwarz inequality with respect to expectations. Finally, putting the results all together completes the proof.
\end{proof}
\begin{lemma} \label{Lemma: Ineq Ind}
	Suppose that $(X_k, y_k)$ are independent for $k=0,1, \cdots N-1$. Then 
	\begingroup\makeatletter\def\f@size{8.5}\check@mathfonts
	\def\maketag@@@#1{\hbox{\m@th\large\mnormalfont#1}}
	$$ \expec{\vnorm{ \sum_{k=0}^{N-1} F_kX_k y_k }^2} \leq \left( \sum_{k=0}^{N-1} \mnorm{F_k}\sqrtp{ \expec{ \mnorm{X_k}^2 } \expec{\vnorm{y_k}^2} } \right)^2.$$\endgroup
\end{lemma}
\begin{proof}
	\begin{align*}
		&\expec{\vnorm{ \sum_{k=0}^{N-1} F_k X_k y_k }^2} \leq 
		\expec{ \left( \sum_{k=0}^{N-1} \mnorm{F_k} \mnorm{X_k} \vnorm{y_k} \right)^2 } \\
		& ~= \sum_{k,l=0}^{N-1} \expec{ \Big( \mnorm{F_k}\mnorm{X_k} \vnorm{y_k}\Big) \Big( \mnorm{F_l} \mnorm{X_l} |\vnorm{y_l}  \Big)} \\
		& ~ \leq \sum_{k,l=0}^{N-1} \mnorm{F_k} \mnorm{F_l} \sqrtp{ \expec{ \mnorm{X_k}^2 \vnorm{y_k}^2} \expec{ \mnorm{X_l}^2 |\vnorm{y_l}^2  } } \\
		& ~ \leq  \left( \sum_{k=0}^{N-1} \mnorm{F_k} \sqrtp{ \expec{ \mnorm{X_k}^2 } \expec{\vnorm{y_k}^2} } \right)^2,
	\end{align*}
	where the first inequality follows from Lemma~\ref{Lemma: Ineq Triangle Sub-multiplicative}, the second inequality follows from applying the Cauchy Schwarz inequality with respect to expectations, and the last one is a result of the mutual independence of $(X_k, y_k)$. 
\end{proof}
\begin{lemma} \label{Lemma: Ineq Ind zero}
	Suppose that $\expec{X_k} = 0$, $\{X_k\}$ has independent increments, i.e. $(X_k,X_l)$ are independent for $k\neq l$, and $(X_k, y_l)$ are independent for $k \geq l$ with $k,l=0,1, \cdots N-1$. Then
	$$ \expec{\vnorm{ \sum_{k=0}^{N-1} F_k X_k y_k }^2} \leq \sum_{k=0}^{N-1}\mnorm{F_k}^2 \expec{\mnorm{X_k}^2} \expec{\vnorm{y_k}^2}.$$
\end{lemma}
\begin{proof}
	\begin{align*}
		&\expec{\vnorm{\sum_{k=0}^{N-1} F_k X_k y_k }^2} \leq
		\expec{ \left( \sum_{k=0}^{N-1} \mnorm{F_k} \mnorm{X_k} \vnorm{y_k} \right)^2 } \\
		& \quad = \sum_{k=0}^{N-1} \mnorm{F_k}^2 \expec{ \mnorm{X_k}^2 \vnorm{y_k}^2 } \\ 
		& \qquad \qquad  + \sum_{\substack{k,l=0 \\ k<l}}^{N-1} \mnorm{F_k} \mnorm{F_l} \mathbb E \Big[ \mnorm{X_k} \vnorm{y_k} \vnorm{y_l} \Big] \mathbb E \Big[\mnorm{X_l} \Big]\\ 
		& \qquad \qquad + \sum_{\substack{k,l=0 \\ k>l}}^{N-1} \mnorm{F_k} \mnorm{F_l} \mathbb E \Big[ \vnorm{y_k} \mnorm{X_l} \vnorm{y_l} \Big] \mathbb E \Big[\mnorm{X_k}\Big]\\
		& \quad = \sum_k^{N-1} \mnorm{F_k}^2\expec{\mnorm{X_k}^2 \vnorm{y_k}^2} +0 + 0\\
		& \quad = \sum_k^{N-1} \mnorm{F_k}^2\expec{\mnorm{X_k}^2} \expec{\vnorm{y_k}^2},
	\end{align*}
	where the first inequality follows by applying Lemma~\ref{Lemma: Ineq Triangle Sub-multiplicative}, and the first equality follows from the independence of $(X_k, y_l)$ when $k > l$ and the fact that $X_k$ has independent increments. The second equality follows because $X_k$ is zero-mean, and the last equality holds because the pair $(X_k,y_k)$ are mutually independent. 
\end{proof}

\subsection{Total \& Quadratic Variations of Deterministic Functions} \label{Section: Variations of M}
Let $\mathcal C$ denote the class of deterministic, matrix-valued functions $M$ that can be decomposed into two parts $M(t) = C(t) + D(t)$, where $C(t)$ is differentiable and $D(t)$ includes all the jumps (or discontinuities) of $M$, i.e.
\begin{equation} \label{Eqn: Function Decomposition}
	M(t) = C(t) + D(t); \quad \text{s.t.} \quad D(t) = \sum_j A_j \mathds 1(t-\tau_j),
\end{equation}	
where $\{A_j\}$ are constant matrices that correspond to the jumps at $\{\tau_j\}$, and $\mathds 1(t)$ is the Heaviside step function centered at zero. Note that if $M$ is a scalar function, $\mathcal C$ boils down to the class of functions with bounded absolute variations.

Define the total and quadratic variations of $M \in \mathcal C$ over the interval $[0,t]$ as
\begin{align*}
	\tv{t}{M} &:= \lim_{N\to \infty} \sum_{k=0}^{N-1} \mnorm{ M(t_{k+1}) - M(t_k) } \\
	\qv{t}{M} &:= \lim_{N\to \infty} \sum_{k=0}^{N-1} \mnorm{ M(t_{k+1}) - M(t_k) }^2,
\end{align*}
respectively, where $\mathcal P_N[0,t]$ (Section~\ref{Section: Partition}) is used to partition the interval $[0,t]$.
\begin{lemma} \label{Lemma: Variations}
	If $M \in \mathcal C$, then $\tv{t}{M}$ and $\qv{t}{M}$ are finite for any finite time $t$.
\end{lemma}
\begin{proof}
	Since $M \in \mathcal C$, we exploit the decomposition in (\ref{Eqn: Function Decomposition}) to write the total variation of $M$ as
	\begin{align*}
		\tv{t}{M} &= \lim_{N\to \infty} \sum_{k=0}^{N-1} \mnorm{ \tilde C_k + \tilde D_k } \\
		& \leq \lim_{N\to \infty} \sum_{k=0}^{N-1} \mnorm{ \tilde C_k} + \lim_{N\to \infty} \sum_{k=0}^{N-1} \mnorm {\tilde D_k } \\
		& = \tv{t}{C} + \tv{t}{D},
	\end{align*} 
	where the notation in Section~\ref{Section: Increments} for the increments is used, i.e. $\tilde C_k := C(t_{k+1}) - C(t_k)$.
	$\tv{t}{C}$ is shown to be finite by exploiting the fact that $C$ is differentiable, i.e.
	\begin{align*}
		\tv{t}{C} = \lim_{N\to \infty} \sum_{k=0}^{N-1} \mnorm{ \frac{\tilde C_k}{\Delta_k} } \Delta_k = \int_0^t \mnorm{\dot C(\tau)} d\tau.
	\end{align*}
	The integral is finite, because $C$ is differentiable and thus $\mnorm{\dot C(t)}$ is finite for finite time. Furthermore, $\tv{t}{D}$ is finite because
	\begin{align*}
		\tv{t}{D} &= \lim_{N\to \infty} \sum_{k=0}^{N-1} \mnorm{ \sum_j A_j \Big(\mathds{1}(t_{k+1} - \tau_j) - \mathds{1}(t_k - \tau_j) \Big) } \\
		&= \sum_j \mnorm{A_j},
	\end{align*}
	where the second equality follows from the fact that the increments of the Heaviside step function are zeros everywhere except at the jumps $\{\tau_j\}$.
	Therefore, $\tv{t}{M}$ is finite over any bounded interval $[0,t]$ with an upper bound given by
	\begin{align*}
		\tv{t}{M} \leq \int_0^t \mnorm{\dot C(\tau)} d\tau + \sum_j \mnorm{A_j}.
	\end{align*}
	Similar reasoning can be carried out to show that $\qv{t}{M}$ is also finite. In fact, using similar arguments we obtain
	\begin{align*}
		\qv{t}{M} & \leq \qv{t}{C} + \qv{t}{D} + 2\lim_{N\to \infty} \sum_{k=0}^{N-1} \mnorm{\tilde C_k} \mnorm{\tilde D_k} \\
		& \leq 0 + \sum_j \mnorm{A_j}^2 + 0.
	\end{align*}
\end{proof}
\subsection{Second Moment of Quadratic Variations} \label{Section: Bounded Quadratic Variation}
The goal of this appendix is to show that the second moment of the quadratic variation of the solutions of (\ref{Eqn: Single SIE Ito}) is finite over finite time. For simplicity, we consider the scalar case with $w=0$, $M_0 = 0$ and $\cov \Gamma = 1$; however the same analysis can be carried out for the general case. Over the partition $\mathcal P_N[0,t]$, (\ref{Eqn: Single SIE Ito}) can be expressed as $ y_k = \sum\limits_{l=0}^{k-1} M(t_k-t_l) y_l \tilde \gamma_l$ and thus the increments can be written as
$$ \tilde y_k = M(\Delta_k) y_k \tilde \gamma_k + \sum_{l=0}^{k-1} \big(M(t_{k+1}-t_l) - M(t_k - t_l)\big) y_l \tilde \gamma_l.$$
Using the inequality $(a+b)^4 \leq 8(a^4 + b^4)$ and the Cauchy Schwarz inequality, we obtain
\begin{align*}
	\expec{\tilde y_k^4}	& \leq 8 M^4(\Delta_k) \expec{y_k^4 \tilde \gamma_k^4} + 8L_M^4 \cov \Delta^2 t_k^2 \expec{\left(\sum_{l=0}^{k-1} y_l^2 \tilde \gamma_l^2\right)^2},
\end{align*}
where $L_M$ is the Lipschitz constant of $M$ and $\cov \Delta$ is defined in Section~\ref{Section: Partition}. Using Lemma~\ref{Lemma: Ineq Ind}, $\expec{\tilde \gamma_k^4} = 3 \Delta_k^2$, and Assumption~\ref{Ass: M} yield the upper bound $\expec{\tilde y_k^4} \leq c(t) \cov \Delta^2$, where $c(t)  = 24 \left(c_M^4(t) + L_M^4 t^4 \right)\sup_{\tau \leq t} \expec{y^4(\tau)} $. Note that $\sup_{\tau \leq t}\expec{y^4(\tau)}$ is shown to be finite in the corollary of \cite[Thm 3.1]{ito1979existence}. Therefore, using the Cauchy-Schwarz inequality with respect to expectations, the second moment of the quadratic variation  over $\mathcal P_N[0,t]$ can be bounded as follows
\begin{align*}
	\expec{\left(\sum_{k=0}^{N-1} \tilde y_k^2 \right)^2}	& \leq  \left( \sum_{k=0}^{N-1}\sqrt{\expec{\tilde y_k^4}}\right)^2 \leq  c(t) \left( \sum_{k=0}^{N-1}\cov \Delta\right)^2.
\end{align*}
Finally, taking the limit as $N\to \infty$ shows that $\expec{\langle y \rangle^2(t)}$ is bounded for finite time $t$.
\end{document}